\documentclass[10pt]{article}
\usepackage{graphicx} % Required for inserting images
\usepackage[pagebackref, colorlinks = true, linkcolor = blue, urlcolor  = blue, citecolor = red, bookmarks, hypertexnames=false]{hyperref}
\usepackage{authblk}

\usepackage[a4paper, margin=1.3in]{geometry}

\usepackage{todonotes}

\usepackage{amssymb}
\usepackage{amsmath}
\usepackage{amsthm}
\usepackage{dsfont}
\usepackage{fdsymbol}
\usepackage{tikz}
\usepackage{tikz-cd}
\usetikzlibrary{shapes,calc,decorations.markings, decorations.pathmorphing, 3d}
\usepackage{mathtools}
\usepackage{physics}
\usepackage{braket}
\usepackage{xcolor}
\usepackage[capitalise]{cleveref}
\usepackage[utf8]{inputenc}
\usepackage[most]{tcolorbox}
\usepackage[skip=5pt]{caption}

\usetikzlibrary{calc}
\usetikzlibrary{fadings}
\usetikzlibrary{shapes.misc, shapes.geometric}
\usetikzlibrary{decorations.pathreplacing,decorations.pathmorphing,angles,quotes}

\newtheoremstyle{newdefinition}{}{}{\normalfont}{}{\bfseries}{}{\newline}
{\thmname{#1} \thmnumber{#2}\thmnote{ (#3)}}

\newtheoremstyle{newplain}{}{}{\itshape}{}{\bfseries}{}{1em}
{\thmname{#1} \thmnumber{#2}\thmnote{ (#3)}}

\newtheoremstyle{newremark}{}{}{\normalfont}{}{\bfseries}{}{1em}
{\thmname{#1}}

\newtheorem{definition}{Definition}[section]

\theoremstyle{newplain}
\newtheorem{theorem}[definition]{Theorem}
\newtheorem{lemma}[definition]{Lemma}
\newtheorem{proposition}[definition]{Proposition}
\newtheorem{corollary}[definition]{Corollary}

\DeclareMathOperator{\Z}{\mathbb{Z}}
\DeclareMathOperator{\C}{\mathbb{C}}
\DeclareMathOperator{\E}{\mathbb{E}}

\DeclareMathOperator{\HH}{\mathcal{H}}

\DeclareMathOperator{\BB}{\mathcal{B}}

\DeclareMathOperator{\FF}{\mathcal{F}}
\DeclareMathOperator{\LL}{\mathcal{L}}

\DeclareMathOperator{\NN}{\mathcal{N}}

\DeclareMathOperator{\PP}{\mathcal{P}}

\DeclareMathOperator{\RR}{\mathcal{R}}
\DeclareMathOperator{\SSS}{\mathcal{S}}

\DeclareMathOperator{\supp}{supp}
\DeclareMathOperator{\diam}{diam}
\DeclareMathOperator{\dist}{dist}
\DeclareMathOperator{\Img}{Im}

\DeclareMathOperator{\gap}{gap}
\newcommand{\identity}{\ensuremath{\mathds{1}}}
\def\lsym{\Lambda} % set of qubits
\DeclareMathOperator{\strset}{{\mathbb{S}}}
\DeclareMathOperator{\plqset}{{\mathbb{P}}}
\newcommand{\bdy}{\partial}
\newcommand{\cobdy}{d}

\def\str{\medwhitestar}
\def\plq{\square}

\numberwithin{equation}{section}

\tikzset{ equation/.style={ baseline={([yshift=-1.5ex]current bounding box.center)} }, torus horizontal/.style = { decoration={ markings, mark=at position 0.5 with { \draw (-2pt,-2pt) -- (2pt,2pt); \draw (2pt,-2pt) -- (-2pt,2pt); }}, decorate}, torus vertical/.style = { decoration={ markings, mark=at position 0.5 with { \draw (-2pt,-2pt) -- (2pt,2pt); \draw (-3pt,-2pt) -- (1pt,2pt); }}, decorate} }

\title{Modified logarithmic Sobolev inequalities for\\ Abelian quantum double models}

 \author[1,2]{Sebastian Stengele}
 \author[3,4]{\'Angela Capel}
 \author[5,6]{Li Gao}
 \author[7]{Angelo Lucia}
 \author[8,9]{David Pérez-García}
 \author[10]{Antonio Pérez-Hernández}
 \author[11]{Cambyse Rouzé}
 \author[1,2]{Simone Warzel}

 \affil[1]{\small Departments of Mathematics and Physics, TU M\"{u}nchen, 85747 Garching, Germany} 
 \affil[2]{\small Munich Center for Quantum Science and Technology,
 80799 M\"{u}nchen, Germany}
 \affil[3]{\small Department of Applied Mathematics and Theoretical Physics, University of Cambridge, Wilberforce Road, Cambridge, CB3 0WA, United Kingdom}
 \affil[4]{\small Fachbereich Mathematik, Universität Tübingen, 72076 Tübingen, Germany}
 \affil[5]{\small School of Mathematics and Statistics, Wuhan University, Wuhan, 430072, China}
 \affil[6]{\small Wuhan Institute of Quantum Technology, Wuhan, 430075, China}
 \affil[7]{\small Dipartimento di Matematica, Politecnico di Milano, 20133 Milano, Italy}
 \affil[8]{\small Departamento de An\'{a}lisis Matemático y Matemática Aplicada, Universidad Complutense de Madrid, 28040 Madrid, Spain}
 \affil[9]{\small Instituto de Ciencias Matemáticas, 28049 Madrid, Spain}
 \affil[10]{\small Departamento de Matem\'{a}tica Aplicada I, Escuela T\'{e}cnica Superior de Ingenieros Industriales, Universidad Nacional de Educación a Distancia, 28040 Madrid, Spain}
 \affil[11]{\small Inria, Télécom Paris - LTCI, Institut Polytechnique de Paris, 91120 Palaiseau, France}

\date{May 18, 2026} % ArXiv might recompile at random dates, changing the date in the pdf, so its better to fix it.
\begin{document}

\maketitle
\begin{abstract}
    We establish rapid mixing for Davies Markov semigroups associated with 2D Abelian quantum double models at any positive temperature. A condition of Dobrushin-Shlosman (DS) type holds at any temperature, and we show that the latter implies a modified logarithmic Sobolev inequality for the Davies Lindbladian. A key step in the argument is to verify a strong martingale condition for the local conditional expectations of the Davies semigroup in the regime of validity of the DS condition.
\end{abstract}

\section{Introduction and results}

Modified logarithmic Sobolev inequalities (MLSI) \cite{kastoryano_QuantumLogarithmicSobolev_2013} are used to derive tight bounds on the mixing times of quantum systems described by a Markovian quantum semigroup. They estimate the relative entropy between an initial quantum state and the dynamics’ fixed point in terms of the entropy production \cite{spohn_Entropyproductionquantum_1978}.
Such MLSIs have only been established in limited settings, namely for the heat-bath dynamics of the generalized depolarizing semigroup~\cite{capel_Quantumconditionalrelative_2018,beigi_Quantumreversehypercontractivity_2020} and that of specific $1$D systems~\cite{bardet_ModifiedLogarithmicSobolev_2021} as well as for Schmidt-generators of 2-local, commuting Hamiltonians~\cite{capel_ModifiedLogarithmicSobolev_2021}. For the Davies dynamics of  $k$-local, $1$D Hamiltonians an MLSI was established at any positive temperature~\cite{bardet_RapidThermalizationSpin_2023,bardet_EntropyDecayDavies_2024} and, in arbitrary dimension, in the case of  2-local, commuting Hamiltonians at high enough temperature~\cite{kochanowski_RapidThermalizationDissipative_2025} as well as for non-interacting Hamiltonians conjugated with IQP circuits  \cite{bergamaschi_QuantumComputationalAdvantage_2024}. Very recently, an MLSI for the Davies dynamics of CSS codes has been proved~\cite{stengele_ModifiedlogarithmicSobolev_2025} under Dobrushin-Shlosman type conditions of the decay of correlations of the Gibbs state, which are known to be optimal in the classical case.

The main goal of this paper is to demonstrate the flexibility of the approach in \cite{stengele_ModifiedlogarithmicSobolev_2025} by adapting it to Abelian quantum double models. 
They constitute a paradigmatic family of two-dimensional, commuting Hamiltonians exhibiting topological order and anyonic excitations~\cite{kitaev_FaulttolerantQuantumComputation_2003}. It is known that the Davies dynamics of these models exhibit fast mixing \cite{komar_NecessityEnergyBarrier_2016, lucia_ThermalizationKitaevsQuantum_2023}, in the sense that the mixing time is at most polynomial in the system size. The more robust notion of rapid mixing, a mixing time at most polylogarithmic in system size, has been an open question. 
In this work, we show that the 2D Abelian quantum double models satisfy an MLSI, and hence rapid mixing, for any positive temperature.

\subsection{Abelian quantum double models}
The model is defined on $N\times N$ squares of the 
usual square lattice $\Z^2$ equipped with periodic boundary conditions, that is, a torus. %Subsets of the torus will, in general, have open boundaries. 
This lattice is composed of vertices, edges, and faces, which are called \emph{plaquettes}. We denote by $\lsym_{N}$ the set of edges of the $ N \times N $ torus. Each edge is assigned an orientation: for simplicity, we assume that all horizontal edges point to the left and vertical edges point downward, cf.~Figure~\ref{fig:basicsquare}. To investigate the dependence of the mixing times on the system size, we fix a family of increasing lattices $\FF=\{\lsym_{N_1}, \lsym_{N_2}, \ldots\}$.

The boundary of each plaquette $p$ is formed by four edges, which we denote by $\bdy p$. According to the orientation, we will denote by $\bdy_{+} p$ (resp. $\bdy_{-} p$) the subset of edges oriented in the counterclockwise (resp. clockwise) direction along $\bdy p$. At each vertex, four incident edges form a \emph{star} $ s$. We denote the set of edges of a star by $\cobdy s$. According to the orientation,  we will denote by $\cobdy_{+} s$ (resp. $\cobdy_{-} s$) the subset of edges that are pointed away from (resp. to) the vertex.

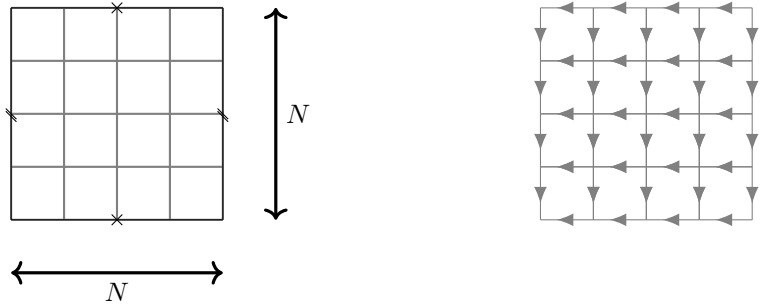
\begin{figure}[ht]
\centering
\begin{tikzpicture}[equation, scale=0.7]

    \draw[step=1.0,gray,thick] (0,0) grid (4,4);
    
    \begin{scope}
        \draw[postaction=torus horizontal] (0,0) -- (4,0);
        \draw[postaction=torus vertical] (4,0) -- (4,4);
        \draw[postaction=torus horizontal] (0,4) -- (4,4);
        \draw[postaction=torus vertical] (0,0) -- (0,4);
        
    \end{scope}
    \draw[<->, very thick] (0,-1) -- node[below]{$N$} (4,-1);
    \draw[<->, very thick] (5,0) --  node[right]{$N$} (5,4);

\begin{scope}[xshift=10cm, decoration={
    markings,
    mark=at position 0.7 with {\arrow[very thick]{latex}}}]
\foreach \x in {1,...,4} 
\foreach \y in {1,...,4}{

\begin{scope}[xshift=\x cm, yshift=\y cm]
\draw[postaction={decorate}, gray] (0,0)  -- (-1,0); 
\draw[postaction={decorate}, gray] (-1,0)  -- (-1,-1); 
\draw[postaction={decorate}, gray] (0,0)  -- (0,-1); 
\draw[postaction={decorate}, gray] (0,-1)  -- (-1,-1); 
\end{scope}
}
\end{scope}

\end{tikzpicture}
\caption{Left: the $N\times N$ lattice on the torus. Right: The orientation of the edges.}\label{fig:basicsquare}
\end{figure}

Let us now fix an arbitrary finite Abelian group $G$ with $ |G| $ elements. We will denote by $\overline{g}$ (instead of $g^{-1}$) the inverse of any $g \in G$, and by $1 \in G$ the identity. The character group of $G$ is the set $\hat G$ of all homomorphisms $\chi:G \to U(1)$, where $U(1)$ denotes the set of complex numbers of modulus one, endowed with the pointwise multiplication. The inverse of a character $\chi$ is denoted by $\overline{\chi}$ and satisfies $\overline{\chi}(g) = \overline{\chi(g)} = \chi(\overline{g})$, and the identity of $\hat{G}$ is the trivial character $\mathbf{1}_{G}$ that assigns the value one to every element of the group. Recall that for finite Abelian groups, $\hat{G}$ is isomorphic to $G$, and by Fourier transform one has the following orthogonality relations
\begin{equation}\label{equa:CharacterOrthogonal}
\frac{1}{|G|}\sum_{g \in G} \chi(g) = 
\begin{cases}
1 & \text{ if } \chi= \mathbf{1}_{G},\\
0 & \text{otherwise}.
\end{cases}
\qquad 
\frac{1}{|G|}\sum_{\chi \in \hat{G}} \chi(g) = 
\begin{cases}
1 & \text{ if } g=1,\\
0 & \text{otherwise}.
\end{cases}
\end{equation}

At each edge $e \in \lsym $ we attach the complex finite-dimensional Hilbert space $\mathcal{H}_{e} = \ell^2(G)$ with $\{ \ket{g} \mid g \in G\}$ denoting an orthonormal basis. The space of local observables is the linear maps
$ \mathcal{B}_{e} = \mathcal{B}(\mathcal{H}_{e}) \equiv \mathcal{M}_{|G|}(\C) $. 
For any subset $V \subseteq \lsym \in \mathcal{F} $, we then consider the tensor product Hilbert space and its observables
\[
    \mathcal{H}_{V} = \bigotimes_{e\in V} \mathcal{H}_e \quad ,\quad \mathcal{B}_{V} = \bigotimes_{e\in V} \mathcal{B}_e.
\]
As usual, we identify (isometrically) $\mathcal{B}_{V}$ with the subspace of $\mathcal{B}_{\lsym}$ consisting of elements with support in $V$ through $Q \mapsto Q \otimes \identity_{\lsym \setminus V}$. 

%The Hamiltonian of the quantum double model with gauge group $G$ on $\Lambda \in \FF$ is defined as
%\begin{equation}\label{equa:quantumDoubleHamiltonian}
%    H_\lsym := - \sum_{s} A_s - \sum_{p} B_p,
%\end{equation}
%\todo{clarify definition for general $ V \subset \Lambda $ see \eqref{eq:fvHam}SST: I moved it up}
%where the first sum runs over all \emph{stars} $s$  and the second sum runs over all \emph{plaquettes} $p$. 

Let $\strset_\Lambda $ and $\plqset_\Lambda $, respectively, stand for the set of stars and plaquettes in $ \Lambda \in \FF $. 
For any $V \subseteq \lsym \in \FF$, we let $\strset_{V}$ be the subset of all stars $s \in \strset$ such that $\cobdy s \cap V \neq \emptyset$. Similarly, $\plqset_{V}$ denotes the subset of all plaquettes $p \in \plqset$ with $\partial p \cap V \neq \emptyset$. The Hamiltonian of the quantum double model on $\HH_V$ is given by
\begin{equation}\label{eq:fvHam}
H_{V} \coloneqq - \sum_{s \in \strset_V} A_{s} - \sum_{p \in \plqset_V} B_{p} \ . 
\end{equation} 
Each operator $A_s$ and $B_p$ is supported in $\cobdy s$ and $\bdy p$, respectively.
% \[
%   \begin{tikzpicture}[equation, scale=0.7]   
%   \draw[step=1,gray,thin] (-1,-1) grid (1,1);
%   \draw[ultra thick, red] (-1,0) -- (1, 0);
%   \draw[ultra thick, red] (0,-1) -- (0,1);
% \end{tikzpicture}
% \hspace{3cm}
%   \begin{tikzpicture}[equation, scale=0.7]   
%   \draw[step=1,gray,thin] (-0.5,-0.5) grid (1.5,1.5);
%     \draw[ultra thick, blue] (0,0) -- (1, 0) -- (1,1) -- (0,1) -- cycle;
% \end{tikzpicture}
% \]
More precisely, for each $g\in G$ and each star $s\in \strset$, we define the (unitary) product operator $A_{s}(g)$, the sum of which defines the star operator:
\begin{equation}  \label{def:Astar}
\begin{tikzpicture}[equation]   
  \draw[step=0.8,gray,thin] (-0.8,-0.8) grid (0.8,0.8);
  \draw[ultra thick, red] (-0.4,0) -- (0.4, 0);
  \draw[ultra thick, red] (0,-0.4) -- (0,0.4);
  \shade[ball color=red] (-0.4,0) circle (0.8ex);
  \shade[ball color=red] (0.4,0) circle (0.8ex);
  \shade[ball color=red] (0,-0.4) circle (0.8ex);
  \shade[ball color=red] (0,0.4) circle (0.8ex);
\end{tikzpicture}
\hspace{1cm}
A_{s}(g) \coloneqq \,
  \begin{tikzpicture}[equation]  
    \draw (0,0) node {$\otimes$};
    \draw (0,0.5) node {$L^{\overline{g}}$};
    \draw (0.5,0) node {$L^{\overline{g}}$};
    \draw (0,-0.5) node {$L^{g}$};
    \draw(-0.5,0) node {$L^{g}$};
  \end{tikzpicture}
  \hspace{1cm}
  A_s \coloneqq \frac{1}{|G|} \sum_{g\in G} A_s(g)  \, .
\end{equation}
Here the \emph{left regular representation} ($L_g$) of $G$ acting on $\ell^2(G)$ is defined as
$ L^{g} \coloneqq \sum_{h \in G} \ketbra{gh}{h} $. 
In~\eqref{def:Astar},  $L^{{g}}$ is applied on the edges of $\cobdy_{+} s$ and  $L^{\overline{g}}$ on the edges of  $\cobdy_{-} s$.
Similarly, for each character $\chi \in \hat{G}$, we define the (unitary) product operator $B_{p}(\chi)$ that acts as $M^{\chi} \coloneqq \sum_{g \in G} \chi(g) \dyad{g} $ on the edges of $\partial_{+}p$,  and as $M^{\overline{\chi}}$ on the edges of $\partial_{-}p$:
\begin{equation}\label{def:Bplaq}
  \begin{tikzpicture}[equation]   
  \draw[step=1.0,gray,thin] (-0.5,-0.5) grid (1.5,1.5);
  \draw[ultra thick, blue, fill=blue!50!white] (0,0) -- (1, 0) -- (1,1) -- (0,1) -- (0,0); 
  %\draw[fill=blue, blue, ultra thick] (0.05,0.05) -- (0.95, 0.05) -- (0.95,0.95) -- (0.05,0.95) -- (0.05,0.05); 
  \shade[ball color=blue] (0.5,0) circle (0.8ex);
  \shade[ball color=blue] (1,0.5) circle (0.8ex);
  \shade[ball color=blue] (0.5,1) circle (0.8ex);
  \shade[ball color=blue] (0,0.5) circle (0.8ex);
\end{tikzpicture}
\hspace{1cm}
  B_p(\chi) = \, 
  \begin{tikzpicture}[equation]  
    \draw (0,0) node {$\otimes$};
    \draw (0,0.5) node {$M^\chi$};
    \draw(-0.5,0) node {$M^\chi$};
    \draw (0,-0.5) node {$M^{\overline{\chi}}$};
    \draw (0.5,0) node {$M^{\overline{\chi}}$};
  \end{tikzpicture}
  \hspace{1cm}
  B_p = \frac{1}{|\hat{G}|}\sum_{\chi \in \hat{G}} B_p(\chi) \, .
\end{equation}
It is straightforward to check that 
\begin{equation}\label{equa:MandLrelation}
M_v^\chi L_v^g = \chi(g) L_v^g M_v^\chi, \quad L_v^gL_v^{g'} = L_v^{gg'}, \quad M_v^{\chi}M_v^{\chi'} = M_v^{\chi\chi'}
\end{equation}
from which it follows that all $A_s(g)$ and $B_p(\chi)$ commute, and thus
\begin{equation*}
    [A_s,B_p]=0, \quad  [A_s, A_{s'}]=0, \quad  [B_p, B_{p'}]=0 \ .
\end{equation*}
%$[A_s(g),B_p(\chi)]=0$ and $[A_s,B_p]=0$ for all $s$,$p$,$g$ and $\chi$.
A more familiar form for the plaquette operators can be obtained by expanding each operator $M^{\chi}$ appearing in $B_{p}(\chi)$ and using the orthogonality relation \eqref{equa:CharacterOrthogonal}:
\[ B_p = 
\sum_{g_{1}, g_{2}, g_{3}, g_{4} \in G} \delta_{1}(g_{1} g_{2} \overline{g}_{3} \overline{g}_{4}) \,\,\,\,
  \begin{tikzpicture}[equation]  
    \draw (0,0) node {$\otimes$};
    \draw (0,0.5) node {$\dyad{g_{1}}$};
    \draw(-1,0) node {$\dyad{g_{2}}$};
    \draw (0,-0.5) node {$\dyad{g_{3}}$};
    \draw (1,0) node {$\dyad{g_{4}}$};
  \end{tikzpicture}
\]
where $\delta_{1}$ is Kronecker's delta  which takes value $1$ at $1 \in G$ and $0$ otherwise. 

\subsection{Davies Lindbladian and thermal equilibrium}

% In our argument, we will show and make use of the following decay of correlations condition, called the DS-condition.
% There exist $K,\xi>0$ such that for any two overlapping rectangles $UV, VW\subseteq \lsym$
% \begin{equation*}
%             \left\| \frac{\Tr_{VW} \hat{\rho}_{UVW}}{\Tr_{V} \hat{\rho}_{UV}} -1 \right\| \leq \epsilon \ .
% \end{equation*}
The Davies dynamics is the standard Markovian model for thermal noise in the weak-coupling regime~\cite{davies_GeneratorsDynamicalSemigroups_1979}. It is a quantum Markov semigroup generated by a Lindbladian, which is constructed from a set of \emph{bare jump operators}. For quantum double models they act on the Hilbert space of edges $e\in\lsym$: 
\[ \SSS_{e}=\{S_{e, i} | i \in I \} \subset \BB(\HH_{e}) ,\]
and are assumed to be 
closed under the adjoint operation, and ergodic in the sense that its commutant is trivial, 
\[ 
\{ S_{e, i}^{\dagger} | i \in I \} = \{ S_{e, i}| i \in I \} \quad \mathrm{and} \quad \SSS_{e}' = \BB(\HH_{e})' \ .
\]
Throughout, we additionally assume translation invariance in the sense that for any pair $e,e'\in \lsym$ of edges with identical orientation, the sets $ \SSS_{e} $ and $ \SSS_{e'} $ agree up to translation on the tensor-product Hilbert space. This 
restriction is necessary for the proof, but not fundamental to the construction. 

%For this paper, we will always assume single-edge, translation-invariant bare jump operators.
%The system Hamiltonian can be diagonalized
%\[ H = \sum_{\lambda} \lambda P_{\lambda}. \]
Let $\{P_\lambda\}_\lambda$ be the spectral projections of the Hamiltonian. Denote by $\Omega$ the set of all Bohr frequencies, namely, the set of all differences $\lambda - \lambda '$ of eigenvalues of $H_\Lambda$. For every operator $Q \in \BB(\HH_{\lsym})$ and $\omega \in \Omega$, we define 
\[ Q(\omega) = \sum_{\lambda, \lambda' \colon \lambda - \lambda' = \omega} P_{\lambda} Q P_{\lambda'}. \]
The Davies generator on $V\subseteq \lsym$ is then defined as
\begin{equation*}
    \LL_{V}^{(\beta)} = \sum_{e\in V} \mathcal{D}_e^{(\beta)} 
\end{equation*}
where for $e \in \lsym$:
\begin{equation}
    \mathcal{D}_e^{(\beta)}(O) = \sum_{i\in I, \omega \in \Omega} h_{e,i}^{(\beta)}(\omega) \left(S_{e,i}(\omega)^\dag O S_{e,i}(\omega) - \frac{1}{2}\left\{ S_{e,i}(\omega)^\dag S_{e,i}(\omega), O\right\} \right)\ ,
\end{equation}
where $\{\cdot,\cdot\}$ denotes the anticommutator. The operators $S_{e,i}(\omega)$ are called jump operators.
The jump rates $h_{e,i}^{(\beta)} (\omega)$ are assumed to be positive, real numbers, which satisfy the detailed balance condition 
\begin{equation}\label{eq:DB}
h_{e,i}^{(\beta)}(-\omega) = h_{e,i}^{(\beta)}(\omega) e^{-\beta \omega} \ .
\end{equation}
 To avoid issues with jump rates vanishing asymptotically at infinity, we assume that there exists $g>0$ such that uniformly in all sites and in the energy differences $ \omega $:
\begin{equation*}
    \inf_{e, i, \omega}\  h^{(\beta)}_{e,i}(\omega) e^{-\beta\omega/2} \geq g \ .
\end{equation*}

The Lindbladian $\LL_\Lambda^{(\beta)} $ generates the quantum Markov semigroup $\exp\big(t\LL_\Lambda^{(\beta)}\big)$ acting on $\BB(\HH_\Lambda)$.
The dual semigroup $\exp\big( t \LL_\Lambda^{(\beta),*}\big) $, taken with respect to the Hilbert-Schmidt scalar product, describes the evolution of states in the Schr\"{o}dinger picture. When clear from the context, we will regularly omit the superscript $ \beta $. 

Let us recall from \cite{davies_GeneratorsDynamicalSemigroups_1979,frigerio_StationaryStatesQuantum_1978} some helpful properties of the Davies Lindbladian $ \LL_V \equiv \LL_V^{(\beta)} $ for any $V \subseteq \lsym$. 
\begin{enumerate}
\item
The semigroup $e^{t \LL_V}$ is unital and completely positive.
The Gibbs state on $ \HH_\Lambda $ at inverse temperature $\beta$,
\begin{equation}
    \rho \equiv \rho_\Lambda^{(\beta)}\coloneqq \frac{e^{-\beta H_\Lambda}}{Z_\Lambda^{(\beta)}} , 
\end{equation}
where $Z_\Lambda^{(\beta)} = \Tr(e^{-\beta H_\Lambda})$ is the partition function, is a fixed point of the semigroup $e^{t \LL_V^*}$, which is unique when $V=\lsym$. 
\item
For all $s\in[0,1]$, $\LL_V$ is self-adjoint with respect to the scalar product
\begin{equation}\label{equa:weightedScalarProd}
\langle A,B\rangle_{\rho, s}:=  \Tr(\rho^{s} A^{\dagger} \rho^{1-s} B ). 
\end{equation}
For $s=1/2$ this is referred to as the KMS scalar product, while for $s=1$ it is the GNS scalar product. 
%It is known that if $\LL$ satisfies the above property for some $s \in [0,1/2) \cup (1/2, 1]$, then it satisfies this condition for every $s \in [0,1]$. 
The self-adjointness derives from~\eqref{eq:DB} and is also referred to as detailed balance.

Moreover, $-\LL_V$ is positive semidefinite and $\mathbb{C} \identity \subset \ker(\LL_V)$. Thus, using a simple diagonalization argument, the following limits exist  
\[ \E_V^\beta = \lim_{t \rightarrow \infty} e^{t \LL_V^{(\beta)} } \quad \mathrm{and} \quad \E^{\beta,*}_V = \lim_{t \rightarrow \infty} e^{t \LL^{(\beta),*}_V} \]
and define the orthogonal projection onto $\ker(\LL^{(\beta)}_V)$ and $\ker(\LL_V^{(\beta),*})$ with respect to the scalar product \eqref{equa:weightedScalarProd}. We will refer to these maps as the (Davies) \emph{conditional expectations}.
\end{enumerate}
\subsection{Mixing times and MLSI}

Our objective is to estimate the rate of convergence, in trace norm, for every initial state $\sigma \in \SSS(\HH_\Lambda)$:
\[ \|e^{t\LL^*_{\lsym}} (\sigma) - \rho\|_{1} . \]
That is, we want to bound the mixing time
\begin{equation*}
    t_{\mathrm{mix}}(\varepsilon) \coloneqq \inf \{t\geq 0 \ | \ \forall \sigma \in \SSS(\HH_{\lsym}):\ \|e^{t\LL^*_{\lsym}} (\sigma) - \rho\|_{1}\leq \varepsilon\}  \ .
\end{equation*}
There are two main approaches to this problem.
\begin{enumerate}
\item[(i)] \emph{Spectral gap}: The spectral gap of $\LL_{\lsym}$ is the difference between the two largest eigenvalues of $\LL_{\lsym}$.  It admits the following variational formula
\[ \gap(\LL_{\lsym}) = \inf_{Q}\frac{-\langle Q, \LL_{\lsym}(Q)  \rangle_{\rho,1}}{\langle Q, Q - \E_{\lsym}(Q)  \rangle_{\rho,1}}, \]
where the infimum is taken over all $Q \in \BB(\HH_{\lsym})$ with the convention that $\frac{0}{0}:=\infty$. It can then be shown that
\[ \|e^{t\LL^*_{\lsym}}(\sigma) - \rho\|_{1} \leq e^{-t \, \gap(\mathcal{L}_\lsym)}  \sqrt{\rho_{\mathrm{min}}^{-1}}. \]
Here $\rho_{\mathrm{min}}$ is the minimal eigenvalue of $\rho$.
If we assume a uniform lower bound on the gap, and 
 since the minimal eigenvalue $\rho_{\mathrm{min}}$ of the Gibbs state $\rho_{\lsym}$ of a local, commuting Hamiltonian is $O(e^{-|\lsym|})$, this yields a bound on the mixing time of order $O(\mathrm{poly}(|\lsym|))$, which is dubbed \emph{fast mixing}.
\item[(ii)] \emph{Modified Logarithmic Sobolev inequality}: The MLSI constant of the generator $\LL_{\Lambda}$ is defined in terms of the relative entropy $D(\sigma \| \sigma' ) = \tr \sigma (\ln \sigma  - \ln \sigma' ) $ between full-rank states $ \sigma, \sigma' $ by 
\[ \alpha(\LL_{\lsym}) = \frac{1}{2} \inf_{\sigma} \frac{-\frac{\dd}{\dd t}D(e^{t\LL^*_{\lsym}}(\sigma)  \| \rho) \big|_{t=0}}{D(\sigma \| \rho)}, \] %\E^*_{\lsym}(\sigma)
where the infimum is taken over all full-rank states $\sigma$ with $D(\sigma \| \rho) \neq 0$. It can be shown that
\[  \|e^{t\LL^*_{\lsym}} (\sigma) - \rho\|_{1} \leq e^{- \alpha(\LL_{\lsym})\, t} \sqrt{\ln(1/\rho_{\mathrm{min}})} \ . \]
Assuming a uniform lower bound on $\alpha$ then yields a bound on the mixing time of order $O(\mathrm{polylog}(|\lsym|))$ and is dubbed \emph{rapid mixing}.
\end{enumerate}

In this paper, we follow the second approach. Our main result is
\begin{theorem}\label{thm:main}
    The Davies Lindbladians of the 2D Abelian quantum double models satisfy an MLSI with a uniformly strictly positive constant for all strictly positive temperatures.
\end{theorem}
%\LG{Shall We add a math statement like: Namely, we prove that for any $\beta>0$, there exists a constant $c_\beta>0$ such that for $\Lambda$ of any size, $\alpha(\LL_{\lsym}^{(\beta)})\ge c_\beta>0 $. SST:added }
Namely, we prove that for any $\beta>0$, there exists an absolute constant $\alpha_\beta>0$ such that for $\Lambda$ of any size, $\alpha(\LL_{\lsym}^{(\beta)})\ge \alpha_\beta>0 $. We hence establish rapid mixing for any 2D Abelian quantum double model. This improves previous results on fast mixing in \cite{komar_NecessityEnergyBarrier_2016, lucia_ThermalizationKitaevsQuantum_2023}. 

\subsection{Comments and outlook}
The proof of the main theorem follows the strategy developed in \cite{stengele_ModifiedlogarithmicSobolev_2025}, in which 
there are two broad steps. 
In the first step, one derives explicit expressions of the Davies conditional expectations $ \E_{R}^\beta $ projecting on the kernel of $ \LL_R^{(\beta)} $, and uses them to establish an approximate factorization 
$\E_{R}^\beta\circ \E_{R'}^\beta\approx \E_{RR'}^\beta $ on overlapping rectangles. This condition, which we call \emph{strong martingale condition}, is then the input to the second step, a multi-scale analysis. The first step, which is the main focus of this paper, is the only part that differs substantially from the strategy in \cite{stengele_ModifiedlogarithmicSobolev_2025}. One key difference is that, in \cite{stengele_ModifiedlogarithmicSobolev_2025} we split the dynamics into star and plaquette Lindbladians, which are treated separately. While this allows for statements about loss of quantum information in the 3D toric code down to zero temperature, it also complicates matters. Since we are only considering the 2D quantum double model, this split is unnecessary, and we omit it for simplicity. However, we do not see any obstacles in splitting, if one wishes to do so. Treating the full Lindbladian as a whole, though, has the added benefit that we can treat slightly more general jump operators. Any finite set of single-edge operators generating the full algebra works, whereas in \cite{stengele_ModifiedlogarithmicSobolev_2025} only $\{X,Z\}$ was considered. 
A second key difference is that restricting to the 2D double models yields a simpler expression of the kernel of $\LL_\Lambda^{(\beta)}$. We give a simplified proof that gets rid of the ``partition by support'' entirely.

Adapting the approach in~\cite{stengele_ModifiedlogarithmicSobolev_2025} to  Abelian quantum double models on D-dimensional lattices is mostly a matter of notation. The DS-condition (Definition~\ref{def:DS}), which requires uniform decay of correlations in the Gibbs state,  would again be optimal for implying rapid mixing. However, the regime of temperatures in which the DS-condition applies would generally only cover the entire high-temperature regime up to a potential phase transition, and not all temperatures as for $ D=2 $ in \cref{thm:main}.

Generalizing this result to non-Abelian quantum double models, even in two dimensions, is not straightforward. 
A key difficulty is that the methods used here rely heavily on the commutativity of the Gibbs state marginals and on the special structure of the dual group in the Abelian case. In the absence of these properties, deriving a tractable expression for the conditional expectations becomes challenging.
\\

The paper is organized as follows. In \cref{sec:equilibrium} we recall some properties of the Gibbs state of quantum double models and show a variant of decay of correlations called the DS-condition. In the next section, \cref{sec:kernel}, we derive a simple expression of the kernel of the Lindbladian. It is given as an intersection of the set of operators with certain support and a commutant of stars and plaquettes.
In \cref{sec:condexp} we use this expression to derive the infinite time limits of the Lindbladian as a combination of pinchings and traces.
In \cref{sec:strongmartingale} we assemble the DS-condition and the explicit expressions to show the strong martingale condition.
Finally, in \cref{sec:MLSI} we sketch a proof of the fact that the strong martingale condition implies a uniformly bounded MLSI constant. %For details, we refer the reader to \cite{stengele_ModifiedlogarithmicSobolev_2025}

\section{Equilibrium}\label{sec:equilibrium}

\subsection{Some preparations}
Let us start with a few helpful properties of the star and plaquette operators, which we will use throughout the proof.
It is easy to see from (\ref{def:Astar}\,-\ref{equa:MandLrelation}) that for every  $p \in \plqset$, star $s \in \strset$, $g,g' \in G$ and $\chi, \chi' \in \hat{G}$ 
\begin{align}\label{equa:}
& A_{s}(g)^{\dagger} = A_{s}(\overline{g}) \quad , \quad A_{s}(g) A_{s}(g') = A_{s}(gg')\\
& B_{p}(\chi)^{\dagger} = B_{p}(\overline{\chi}) \quad , \quad B_{p}(\chi) B_{p}(\chi') = B_{p}(\chi \chi') .  \notag
\end{align}

In addition to the usual star and plaquette operators, we will make use of the following generalized ones:
\begin{equation}\label{def:genAB}
    A_s(\chi) \coloneqq \frac{1}{|G|}\sum_{g\in G} \chi(g) A_s(g) \quad , \quad  B_p(g) \coloneqq \frac{1}{|\hat{G}|}\sum_{\chi\in \hat{G}} \chi(g) B_p(\chi)
\end{equation}
It is not hard to check that these operators are commuting orthogonal projections, and moreover, they satisfy
\begin{align}
& A_{s}(\chi)^{\dagger} = A_{s}(\chi) \quad , \quad  A_{s}(\chi) A_{s}(\chi') = A_{s}(\chi) \delta_{\chi, \chi'} \quad , \quad \sum_{\chi \in \hat{G}}A_s(\chi) = \identity, \\
& B_{p}(g)^{\dagger} = B_{p}(g) \quad, \quad B_{p}(g) B_{p}(g') = B_{p}(g) \delta_{g,g'} \quad , \quad \sum_{g\in  G}B_p(g) = \identity.  \notag
\end{align}
Thus, each family $(A_{s}(\chi))_{\chi}$ and $(B_{p}(g))_{g}$ forms a resolution of the identity. The local terms of the quantum double Hamiltonian \eqref{eq:fvHam} correspond to
\[ A_{s} = A_{s}(\mathbf{1}_{G}) \quad , \quad B_{p} = B_{p}(1). \]
Let us observe the following useful relations. For every star $s$, edge $e$ and characters $\chi, \chi' \in \hat{G}$ we have
\begin{equation}
 M_e^\chi A_s(\chi') M_e^{\overline{\chi}} = \begin{cases}
    A_s(\chi' \chi) & e\in \cobdy_+ s\\
    A_s(\chi' \overline{\chi}) & e\in \cobdy_- s\\
    A_s(\chi') & e\notin \cobdy s
    \end{cases}
\end{equation}
Analogously, for every plaquette $p$, edge $e$ and group elements $g,h \in G$ we have
\begin{equation}
L_e^g B_p(h) L_e^{\overline{g}} = \begin{cases}
    B_p(h \overline{g}) & e\in \bdy_+ p\\
    B_p(h g) & e\in \bdy_- p\\
    B_p(h) & e\notin \bdy p
    \end{cases}
\end{equation}

% For any $R \subset \lsym$, let $\strset_{R}$ be the subset of all $s \in \strset$ such that $\cobdy s \cap R \neq \emptyset$, and let $\plqset_{R}$ be the subset of all plaquettes  $p \in \plqset$  such that $\partial p \cap R \neq \emptyset$. We denote
% \begin{equation}\label{eq:fvHam}
% H_{R} = - \sum_{s \in \strset_R} A_{s} - \sum_{p \in \plqset_R} B_{p}. 
% \end{equation} 

%
The marginals of the Gibbs state associated with subsets play a prominent role in our analysis.
We mostly restrict attention to star- and plaquette-connected sets. 
\begin{definition}
We say that $R\subseteq \lsym$ is \emph{star-connected} if, for any $e,e'\in R$, there exists a path $e=e_0, s_1,e_1,\ldots, e_{n-1},s_n,e_n=e'$ of edges and stars satisfying $e_i\in R$ and $e_{i-1},e_{i}\in \cobdy s_i$ for all $i=1,\ldots,  n$. Analogously, we say that $R$ is \emph{plaquette-connected} if, for any $e,e'\in R$ there exists a path $e=e_0, p_1, e_1, \ldots, e_{n-1}, p_n, e_n=e'$ with $e_i\in R$ and $e_{i-1},e_{i}\in \bdy p_i$ for all $i=1,\ldots,  n$.  
\end{definition}
%\LG{One may consider add a figure to illustrate the definition above. Of course, it's not necessary if short in time.}
The following is an important property of partial traces of the Gibbs factor over such subsets.
\begin{proposition}\label{prop:marginaldoublecommutant}
 Let $R \subseteq R'\subseteq \lsym$ where $R$ is star- and plaquette-connected. Then,
    \begin{equation*}
        \Tr_{R}(e^{-\beta H_{R'}}) \in \{A_s(\chi), B_p(h)\}_{\chi\in \hat{G}, h\in G, s\in \strset_{R'}, p\in\plqset_{R'}}^{''},
    \end{equation*}
    where $\{\cdot\}^{''}$ denotes the bicommutant. 
    
    In particular, any two marginals of the Gibbs state commute. Moreover, if $R'=R$, we have \begin{equation}\label{equa:marginalcloseIdentity}
    \kappa_{R} \left( 1+ \epsilon_{R} \right)^{-2} \, \identity \leq \Tr_{R}\left( e^{-\beta H_{R}}\right) \leq \kappa_{R}\left( 1+ \epsilon_{R} \right)^2  \identity  
    \end{equation}
with
\begin{align*}
\kappa_{R} \coloneqq |G|^{|R|} (1+\gamma_{\beta})^{|\strset_R| + |\plqset_R|}  \quad & , \quad \gamma_{\beta} \coloneqq \frac{e^{\beta}-1}{|G|}, \\ 
\epsilon_{R} \coloneqq e^{\beta}|G|\left(\tfrac{\gamma_{\beta}}{1+\gamma_{\beta}}\right)^{m_R} \quad & , \quad m_{R} \coloneqq \min\{ |\strset_{R}|, |\plqset_{R}| \} . 
\end{align*}   
    
\end{proposition}
\begin{proof}
    Since all star and plaquette operators commute and $H_{R'}-H_R$ has no support on $R$, we can rewrite 
    \begin{equation*}
         \Tr_{R}(e^{-\beta H_{R'}}) =  e^{-\beta (H_{R'}-H_R)}\Tr_{R}(e^{-\beta H_{R}}) \ .
    \end{equation*}
    As the star and plaquette operators are projections, the first factor is an element of the algebra of stars and plaquettes connected to $R'$. Following  \cite[Theorem 40]{lucia_SpectralGapBounds_2025} the second term admits an explicit expression 
    \begin{equation}\label{equa:explicitMarginalQDM} 
    \Tr_{R}\left( e^{-\beta H_{R}}\right) = \kappa_{R} \left( \left(1- a_{R}\right) \identity + a_{R}  |G| A_{\strset_R}  \right) \left( \left(1- b_{R}\right) \identity + b_{R} |G|  B_{\plqset_R}  \right) 
\end{equation}
where $\kappa_{R}$ and $\gamma_{\beta}$ are as in the statement of the proposition, 
    \[  a_{R} = \left(\tfrac{\gamma_{\beta}}{1+\gamma_{\beta}}\right)^{|\strset_R|} \quad , \quad b_{R}=\left(\tfrac{\gamma_{\beta}}{1+\gamma_{\beta}}\right)^{|\plqset_R|} ,  \]
and the operators in the right side  are defined by
\[ A_{\strset_R} := \frac{1}{|G|} \sum_{g \in G} \prod_{s \in \strset_{R}} A_{s}(g) \quad \text{and} \quad B_{\plqset_R} := \frac{1}{|G|} \sum_{\chi \in \hat{G}} \prod_{p \in \plqset_{R}} B_{p}(\chi) .\]
The proof of~\eqref{equa:explicitMarginalQDM}  relies on the fact that these operators
are mutually commuting orthogonal projectors.

     Consequently, $\Tr_{R}( e^{-\beta H_{R}})$ is an element of the Abelian algebra generated by
    \begin{equation*}
        \left\{A_s(g), B_p(\chi) \mid s\in \strset_{R}, p\in \plqset_{R}, g\in G, \chi\in \hat{G}\right\}\ .
    \end{equation*}
    Note that these operators are the building blocks of the actual stars~\eqref{def:Astar} and plaquettes~\eqref{def:Bplaq} and should not be confused with $A_s(\chi)$ and $B_p(h)$.
    However, the latter span an algebra containing the former, since:
    \begin{equation}
        \sum_{\chi\in \hat{G}} \chi(\overline{g_0}) A_s(\chi) = \frac{1}{|G|} \sum_{g\in G} A_s(g) \sum_{\chi \in \hat{G}} \chi(g\overline{g_0}) = A_s(g_0)
    \end{equation}
    for any $g_0\in G$ and analogously for plaquettes.
    Thus, $ \Tr_{R}(e^{-\beta H_{R'}})$ is an element of the algebra generated by
    \begin{equation*}
         \{A_s(\chi), B_p(h)\}_{\chi\in \hat{G}, h\in G, s\in \strset_{R'}, p\in\plqset_{R'}} .
    \end{equation*}
    The first claim then follows from the bicommutant theorem.
    
    The fact that any two marginals commute follows from the fact that all $A_s(\chi), B_p(h)$ commute.

     For a proof of \eqref{equa:marginalcloseIdentity}, we note that it is straightforward  from~\eqref{equa:explicitMarginalQDM} and $|\strset_{R}|, |\plqset_{R}| \geq m_{R}$ that
     \[ \kappa_{R}(1-c_{R})^2 \cdot \identity \leq \Tr_{R}\left( e^{-\beta H_{R}}\right) \leq \kappa_{R}(1+|G|c_{R})^2\cdot \identity \quad \text{ where }c_{R}:=\left( \frac{\gamma_{\beta}}{1+\gamma_{\beta}}\right)^{m_{R}}. \]
     Finally, observe that
    \[ \frac{1}{1-c_{R}} = 1+\frac{\left(\tfrac{\gamma_{\beta}}{1+\gamma_{\beta}}\right)^{m_{R}}}{1-\left(\tfrac{\gamma_{\beta}}{1+\gamma_{\beta}}\right)^{m_{R}}} \leq 1+\frac{\left(\tfrac{\gamma_{\beta}}{1+\gamma_{\beta}}\right)^{m_{R}}}{1-\left(\tfrac{\gamma_{\beta}}{1+\gamma_{\beta}}\right)} = 1+(1+\gamma_{\beta}) \left(\tfrac{\gamma_{\beta}}{1+\gamma_{\beta}}\right)^{m_{R}}\leq 1+\epsilon_{R}, \]
    and $1+|G|c_{R} \leq 1+ \epsilon_{R}$. 
\end{proof}

\subsection{Marginals and the DS-condition}
The main result in this section concerns the decay of correlations of the Gibbs state. For our purposes, it is sufficient to consider this decay on rectangles, which we characterize next.

\begin{definition}
A rectangular region $\mathfrak{R}$ with side lengths $0<l_1,l_2< N$ is a set 
\begin{equation*}
\mathfrak{R}=\{a,a+1,\ldots , a+l_1\}\times \{b, b+1, \ldots b+l_2\}\subseteq \Z_N\times \Z_N
\end{equation*}
for some $a,b\in \Z_N$, where we allow periodic boundaries.
A rectangle $R\subseteq \lsym$ is then the set of all edges with both ends in a rectangular region $\mathfrak{R}$. 
Note that, by definition, we only consider rectangles that have a non-zero width and height.
Observe also that any rectangle is both star- and plaquette-connected. 
For any set $V\subseteq R$, we define its (outer) diameter $\diam(V)$ to be the side length of the smallest square that fully contains $V$. Its inner diameter $\diam_{-}(V)$ is the side length of the smallest square that is fully contained in $V$ and $0$ if no such square exists.
\end{definition}

The following is the analogue of Dobrushin and Shlosman's uniform decay condition, characterizing what is known as the Dobrushin-Shlosman high-temperature regime \cite{dobrushin_CompletelyAnalyticalInteractions_1987}. In the absence of a static phase transition, this regime extends down to arbitrary non-zero temperature. As Theorem~\ref{thm:DS} shows, this is the case for the 2D Abelian quantum double model. The condition involves the marginals
    \begin{equation} \label{eq:rhohatdef}
        \hat{\rho}_R \coloneqq \frac{e^{-\beta H_\lsym}}{\Tr_R e^{-\beta H_\lsym}}
    \end{equation}
    of the Gibbs state on a set $R\subseteq \lsym$.

\begin{definition}[DS-condition]\label{def:DS}
    The Gibbs state of a family $\FF$ of double models  satisfies the DS-condition at inverse temperature $\beta$ if there exist a minimal system $\lsym_0\in \FF$, a minimal distance $d_0\in[2,\infty)$,  and constants $K,\xi >0$ such that for any $\lsym \in \FF$ with $\lsym \supseteq \lsym_0$ and any two overlapping rectangles $UV,VW\subseteq \lsym$ with $\dist(U,W)\geq d_0$ and $\diam_{-}(U), \diam_{-}(W)\geq1$, the marginals of the Gibbs state satisfy:
    \begin{equation}\label{eq:DScond}
        (1-K e^{-\xi \dist(U,W)} )\Tr_{V} \hat{\rho}_{UV} \leq \Tr_{VW} \hat{\rho}_{UVW} \leq  (1+K e^{-\xi \dist(U,W)})\Tr_{V} \hat{\rho}_{UV} .
    \end{equation}
    % \begin{equation}
    %     \left\| \frac{\Tr_{VW} \hat{\rho}_{UVW}}{\Tr_{V} \hat{\rho}_{UV}} -1 \right\| \leq K e^{-\xi \dist(U,W)} \ .
    % \end{equation}
\end{definition}
To simplify the argument in this paper, we use a different but equivalent condition compared to \cite{stengele_ModifiedlogarithmicSobolev_2025}, where we defined
\begin{equation*}
    \left\| \frac{\Tr_{VW} \hat{\rho}_{UVW}}{\Tr_{V} \hat{\rho}_{UV}} -1 \right\| \leq K e^{-\xi \dist(U,W)} \ .
\end{equation*}
The equivalence derives from the fact that all marginals commute.

\begin{theorem}\label{thm:DS}
    The Gibbs state of the 2D Abelian quantum double model satisfies the DS-condition for all positive temperatures $\beta \in [0,\infty)$. More specifically, for any two overlapping rectangles $UV$ and $VW$ with $\dist(U,W) \geq 2 $   and $\diam_{-}(U), \diam_{-}(W)\geq1$ 
    it holds that
    % \begin{equation*}
    %    \left(1 + K \left( \tfrac{\gamma_{\beta}}{1+\gamma_{\beta}}\right)^{\dist(U,W)} \right)^{-1}
    %    \Tr_{V} \hat{\rho}_{UV} 
    %    \leq \Tr_{VW} \hat{\rho}_{UVW} 
    %    \leq   \left(1+  K \left( \tfrac{\gamma_{\beta}}{1+\gamma_{\beta}}\right)^{\dist(U,W)} \right)
    %    \Tr_{V} \hat{\rho}_{UV} 
    % \end{equation*}
     \begin{equation*}
       \left(1 + K e^{-\xi \dist(U,W)} \right)^{-1}
       \Tr_{V} \hat{\rho}_{UV} 
       \leq \Tr_{VW} \hat{\rho}_{UVW} 
       \leq   \left(1+  K e^{-\xi \dist(U,W)} \right)
       \Tr_{V} \hat{\rho}_{UV} 
    \end{equation*}
    with
    \begin{equation*}
        K = 2^{10} e^{10\beta} |G|^{10} \quad \text{ and } \quad e^{-\xi} = \frac{\gamma_\beta}{1+ \gamma_\beta} = \frac{e^{\beta}-1}{e^{\beta}+|G|-1} \ .
    \end{equation*}
    %\[ \left\| \frac{\Tr_{VW} \hat{\rho}_{UVW}}{\Tr_{V} \hat{\rho}_{UV}}  - \identity \right\| \leq  2^{10} e^{\beta} |G|\left( \tfrac{\gamma_{\beta}}{1+\gamma_{\beta}}\right)^{\dist(U,W)} \quad \text{ where } \gamma_{\beta} = \frac{e^{\beta}-1}{|G|}. \]
\end{theorem}
\begin{proof}
A distance of at least $2$ between $U$ and $W$ implies that $UVW$ is a rectangle and that no star or plaquette is connected to both $U$ and $W$. Thus, $H_{UV}+H_{VW}=H_{UVW}+H_{V}$.
The argument follows the ideas of \cite[Theorem 42]{lucia_SpectralGapBounds_2025}. We factorize
\begin{align*} 
\Tr_{VW} \hat{\rho}_{UVW} 
& = \Tr_{VW}(e^{-\beta H_\lsym}) \Tr_{UVW}(e^{-\beta H_\lsym})^{-1}\\
& = \Tr_{VW}(e^{-\beta H_{VW}}) e^{-\beta (H_\lsym - H_{VW})} e^{\beta (H_\lsym - H_{UVW})} \Tr_{UVW}(e^{-\beta H_{UVW}})^{-1}\\
& = e^{\beta(H_{VW} - H_{UVW})} \Tr_{VW}(e^{-\beta H_{VW}}) \Tr_{UVW}(e^{-\beta H_{UVW}})^{-1}\ ,
\end{align*}
and similarly
\begin{align*} 
(\Tr_{V} \hat{\rho}_{UV})^{-1} 
& =  \Tr_{UV}(e^{-\beta H_\lsym}) \Tr_{V}(e^{-\beta H_\lsym})^{-1}\\
& = \Tr_{UV}(e^{-\beta H_{UV}}) e^{-\beta (H_\lsym - H_{UV})}  e^{\beta (H_\lsym - H_{V})} \Tr_{V}(e^{-\beta H_{V}})^{-1}  \\
& = e^{\beta(H_{UV} - H_{V})}  \Tr_{UV}(e^{-\beta H_{UV}}) \Tr_{V}(e^{-\beta H_{V}})^{-1}.
\end{align*}
Multiplying both expressions and using that
\[H_{UV} + H_{VW} - H_{UVW} - H_{V} = 0\ ,\]
we deduce the following decomposition, where all factors commute
\begin{equation}\label{equa:firstdecompDS}
\frac{\Tr_{VW} \hat{\rho}_{UVW}}{\Tr_{V}(\hat{\rho}_{UV})} = \Tr_{VW}(e^{-\beta H_{VW}}) \Tr_{UVW}(e^{-\beta H_{UVW}})^{-1} \Tr_{V}(e^{-\beta H_{V}})^{-1} \Tr_{UV}(e^{-\beta H_{UV}}) \ . 
\end{equation}
Next, we estimate the four factors using \eqref{equa:marginalcloseIdentity}. Note that for $R \in \{ UV,VW, UVW\}$, we have $m_{R} \geq \dist(U,W)$, and therefore
\begin{equation}\label{equa:firstEstimateDS} 
\kappa_{R} \left( 1+ e^{\beta} |G|e^{-\xi\dist(U,W)} \right)^{-2} \, \identity \leq \Tr_{R}\left( e^{-\beta H_{R}}\right) \leq \kappa_{R}\left( 1+ e^{\beta} |G|e^{-\xi\dist(U,W)} \right)^2  \identity 
\end{equation}
For $R=V$ we distinguish two cases. If $V$ is a rectangle, then the above estimates \eqref{equa:firstEstimateDS}  also work in this case. 
However, if $V$ is the union of two rectangles $V=V_{1}V_{2}$, we observe $\dist(V_1, V_2)\geq 2$ since the inner diameters of $U$ and $W$ are both at least $1$. Thus $\strset_{V_{1}} \cap \strset_{V_{2}} = \plqset_{V_{1}} \cap \plqset_{V_{2}} = \emptyset$.
In this case, we split
\[ \Tr_{V}(e^{-\beta H_{V}}) = \Tr_{V_{1}}(e^{-\beta H_{V_{1}}}) \Tr_{V_{2}}(e^{-\beta H_{V_{2}}}),\]
and apply \eqref{equa:firstEstimateDS}  on each of the two factors, which yields
\begin{equation}\label{equa:secondEstimateDS}
\kappa_{V} \left( 1+ e^{\beta} |G|e^{-\xi\dist(U,W)}\right)^{-4} \, \identity \leq \Tr_{V}\left( e^{-\beta H_{V}}\right) \leq \kappa_{V}\left( 1+ e^{\beta} |G|e^{-\xi\dist(U,W)} \right)^4 \identity 
\end{equation}
where we used that $\kappa_{V} = \kappa_{V_{1}} \cdot \kappa_{V_{2}}$. Applying \eqref{equa:firstEstimateDS} and \eqref{equa:secondEstimateDS} on \eqref{equa:firstdecompDS}, we obtain
\[ 
\left( 1+ e^{\beta} |G|e^{-\xi \dist(U,W)} \right)^{-10} \identity
\leq \frac{\kappa_{UVW} \cdot \kappa_{V}}{\kappa_{UV} \cdot \kappa_{VW}} 
\frac{\Tr_{VW}( \hat{\rho}_{UVW})}{\Tr_{V}(\hat{\rho}_{UV})} 
\leq \left( 1+ e^{\beta} |G|e^{-\xi \dist(U,W)} \right)^{10} \identity. \]
We now use that $\kappa_{UVW} \kappa_{V} = \kappa_{UV} \kappa_{VW}$, which follows from the elementary relations 
\begin{align*}
|UV| + |VW| & = |UVW| + |V|,\\
 |\strset_{UV}| + |\strset_{VW}| & = |\strset_{UVW}| + |\strset_{V}|,\\
 |\plqset_{UV}| + |\plqset_{VW}| & = |\plqset_{UVW}| + |\plqset_{V}|,
\end{align*}
which in turn are consequences of the fact that $\dist(U,W) \geq 2$, so that no plaquette or star intersect both $U$ and $W$. 
We finally apply the elementary inequality 
\[ (1+e^{\beta}|G| e^{- \xi \dist(U,W)})^{10} \leq 1 +2^{10} e^{10 \beta} |G|^{10} e^{-\xi \dist(U,W)} \] to arrive at
\begin{equation*}
 \left(1+ 2^{10} e^{10 \beta} |G|^{10} e^{-\xi\dist(U,W)}\right)^{-1} \identity \leq \frac{\Tr_{VW} \hat{\rho}_{UVW}}{\Tr_{V} \hat{\rho}_{UV}} \leq \left(1+ 2^{10} e^{10 \beta} |G|^{10}e^{-\xi\dist(U,W)}\right) \identity. 
\end{equation*}
The proof is completed by using $( 1 + a)^{-1} \geq 1 - a $ for any $ a \geq 0 $. 
\end{proof}

\section{Kernel of the Lindbladian}\label{sec:kernel}
One key insight that enables our result is an explicit expression for the kernel of the local Davies Lindbladians.
Its specific form allows us to write down the infinite-temperature conditional expectation explicitly as a projection onto this kernel. 
As we will show later, this projection factorizes, giving an exact tensorization of the relative entropy.
For finite temperature, the factorization is no longer exact and depends on the correlations of the Gibbs state as expressed in the DS-condition.

We start by expressing the kernels of the local Davies Lindbladians as simple commutants: 
\begin{theorem}\label{theorem:kernelNice}
    Let $R\subseteq \lsym$ be a star- and plaquette-connected set of size at least $2$ and let $\LL_R$ be the Davies Lindbladian of $H$ restricted to $R$. Then
    \begin{equation}\label{eq:kernelNice}
        \ker(\LL_R) = \identity_R \otimes \BB(\HH_{\lsym \setminus R}) \cap \left\{ A_s , B_p \right\}_{s\in \strset_R, p \in \plqset_R}' \, .
    \end{equation}
\end{theorem}
We will give the proof at the end of this section.
It is helpful to slightly rewrite this expression. Namely, one can replace the stars and plaquettes in the commutant by their generalized versions defined in~\eqref{def:genAB}.
\begin{corollary}\label{cor:kernelNice2}
    The kernel of $\LL_R$ can also be written as:
    \begin{equation}\label{eq:kernelNicefull}
        \ker(\LL_R) = \identity_R \otimes \BB(\HH_{\lsym \setminus R}) \cap \left\{ A_s(\chi),  B_p(g) \right\}_{\substack{s\in \strset_R, p \in \plqset_R\\ \chi \in \hat{G}, g\in G}}' \, .
    \end{equation}
\end{corollary}
\begin{proof}
    Comparing the right-hand sides of \eqref{eq:kernelNice} and \eqref{eq:kernelNicefull}, the inclusion $\supseteq$ is immediate.
    For the other inclusion $\subseteq$, we pick $O\in \ker(\LL_R)$. By \eqref{eq:kernelNice}, $O$ commutes with any $A_s$ with $s\in \strset_R$ and any $M_e^\chi$ with $e\in R$. 
    Since $s$ is supported on $R$, we can pick $e\in R \cap \cobdy s$. Then, $O$ also commutes with $M_e^\chi A_s M_e^{\overline{\chi}}$. This is either $A_s(\chi)$ or $A_s(\overline{\chi})$, depending on the orientations. Repeating this for all $\chi\in \hat{G}$ and all stars, we see that $O$ commutes with all $A_s(\chi)$. For plaquettes, the argument is identical.
\end{proof}

Partial traces of stars and plaquettes turn out to be a useful tool in the proof of \cref{theorem:kernelNice}, since they have the following simple expression. 
\begin{lemma}\label{lemma:traceoperator}
For every star or plaquette operator $O \in \{ A_s , B_p \}$ and every set $R\subseteq \lsym$, its partial trace over $R$ evaluates to
\begin{equation}
    \Tr_R(O) = \begin{cases}
        \frac{d_R}{|G|} \identity_{\Lambda \setminus R}, &\supp(O)\cap R \neq \emptyset\\
        d_R O, &\mathrm{else},
    \end{cases}
\end{equation}
where $d_R$ is the dimension of $\HH_R$.
\end{lemma}
\begin{proof}
    We present the star case, plaquettes are analogous. If $A_s$ is not supported on $R$ the statement is trivial. Otherwise, by definition:
    \begin{equation}
        \Tr_R(A_s) = \frac{1}{|G|} \sum_{g\in G} \Tr_R( A_s(g)) \ .
    \end{equation}
    Let $e\in R\cap \supp(A_s)$.  To compute the remaining partial trace, we use the definition~\eqref{def:Astar} of $ A_s(g) $  and the fact that
    \[
    \Tr_e L_e^g = \begin{cases}
        |G| , & g = 1 \\
        0 , &\mathrm{else.}
    \end{cases}
    \]
    Since $ L_e^1= \identity_e $, the claim thus follows. 
\end{proof}

We now present the first step in rewriting the kernel of the Lindbladian. We use the well-known fact that the kernel of a Lindbladian is the commutant of its jump operators. Since the bare jump operators generate the full local algebra by construction, a similar statement can be made for their Fourier transforms.
\begin{lemma}\label{lemma:kernelNice}
For every subset $R \subseteq \lsym$
\begin{equation*}
\ker{(\LL_{R})} = \{S(\omega) | S \in \SSS_{e},  \omega \in \Omega, e \in R \}'
=
\{Q_{R}(\omega) | Q_{R} \in \BB(\HH_{R})\,, \, \omega \in \Omega \}'.
\end{equation*}
\end{lemma}
\begin{proof}
The kernel of a Lindbladian with a full-rank fixed point is known \cite[Thm. 7.2]{wolf_QuantumChannelsOperations_2012} to be given by the commutant of the set of jump operators and their adjoints. This is the first equality. 
Let us denote by $\SSS_{R}$ the union of all $\SSS_{e}$ with $e \in R$. Then, $\SSS_{R}$ is closed under the adjoint operation and ergodic in the sense that 
\[ \SSS_{R}' = \bigcap_{e \in R} \SSS_{e}' = \identity_{R} \otimes \BB(\HH_{\lsym \setminus R})  \equiv \BB(\HH_{\lsym \setminus R}).\]
Therefore, by the bicommutant theorem 
\[ Alg(\SSS_R)  = \BB(\HH_{R}). \]
This means that every $Q_{R} \in \BB(\HH_R)$ can be written as a linear combination of products of elements of $\SSS_{R}$. Thus, by Proposition \ref{Prop:productFourierCoefficients},
\[ \{ Q_{R}(\omega) | Q_{R} \in \BB(\HH_{R}), \omega \in \Omega \} \subseteq Alg(\{ S(\omega) | S \in \SSS_{e}, \omega \in \Omega, e \in R \}), \]
and this yields that
\[ \{ S(\omega) | S \in \SSS_{e}, \omega \in \Omega, e \in R \}' \subseteq \{ Q_{R}(\omega) | Q_{R} \in \BB(\HH_{R}), \omega \in \Omega \}'. \]
The reversed inequality is trivial.
\end{proof}

We have the following useful observation regarding the operator Fourier transform $Q(\omega)$, which is a discrete version of the convolution theorem. 
\begin{proposition}\label{Prop:productFourierCoefficients}
Let $Q, Q' \in \BB(\HH_\lsym)$, then for every $\omega \in \Omega$
\[ (Q Q')(\omega) = \sum_{\nu, \nu' \in \Omega \colon \nu + \nu'=\omega} Q(\nu) Q'(\nu'). \]
\end{proposition}
\begin{proof}
Using the decomposition of unity by the spectral projections $ P_\lambda $ of $ H_\Lambda $ one checks
\begin{align*} 
(Q  Q')(\omega) 
& = \sum_{\lambda, \lambda' \colon \lambda - \lambda' = \omega} P_{\lambda} Q Q' P_{\lambda'} = \sum_{\lambda, \lambda' \colon \lambda - \lambda' = \omega} \sum_{\mu} P_{\lambda} Q P_{\mu} Q' P_{\lambda'} \\
& = \sum_{\lambda, \lambda' \colon \lambda - \lambda' = \omega} \sum_{\mu, \mu'} (P_{\lambda} Q P_{\mu}) (P_{\mu'} Q' P_{\lambda'}) = \sum_{\lambda, \lambda', \mu, \mu' \colon (\lambda-\mu) + (\mu' -\lambda') = \omega}  (P_{\lambda} Q P_{\mu}) (P_{\mu'} Q' P_{\lambda'})\\
& = \sum_{\nu, \nu' \in \Omega \colon \nu + \nu' = \omega} Q(\nu) Q'(\nu')
\end{align*}
which finishes the proof.
\end{proof}

We are now ready to prove the main theorem in this section.
\begin{proof}[Proof of \cref{theorem:kernelNice}]
    We start with a proof of the inclusion $\supseteq$. 
    If $O \in \BB(\HH_{\lsym})$ commutes with all $A_s$ and $B_p$ for $s\in \strset_R$ and $p\in \plqset_R$, it also commutes with $H_R$.
    If $O$ in addition commutes with every $Q_{R}$, it commutes with
    \[ e^{\beta H_R}Q_{R} e^{-\beta H_R} = e^{\beta H_\Lambda} Q_{R} e^{-\beta H_\Lambda} \]
    for every $\beta \in \mathbb{R}$, since $H_\Lambda-H_R$ is not supported on $R$ and commutes with $Q_R$. Taking a suitable filter function $f_{\omega,\Lambda}(\beta)$, we deduce that $O$ commutes with
    \[ Q_{R}(\omega) = \int_{\mathbb{R}}f_{\omega,\Lambda}(\beta) (e^{\beta H_R} Q_{R} e^{-\beta H_R}) d \beta \]
    for every $\omega \in \Omega$. As a consequence, $O$ belongs to $\ker{(\LL_{R})}$.

    To simplify notation for the proof of the reverse inclusion $\subseteq$, let us first define the following equivalence relation on $\BB(\HH)$: $O\cong \tilde{O}$ if and only if $O= \tilde{O}+z\identity$ for some $z\in \C$. Since commutation relations are preserved under the addition of identity, we have $\{O\}'=\{\tilde{O}\}'$ whenever $O\cong \tilde{O}$.

    Any $O \in \ker{(\LL_{R})}$ commutes  by Lemma~\ref{lemma:kernelNice} with every $Q_{R}(\omega)$ with $Q_{R} \in \BB(\HH_{R})$ and $\omega \in \Omega$. In particular, $O$ commutes with
    \[ Q_{R} = \sum_{\omega \in \Omega}Q_{R}(\omega) \]
    and thus $O \in \identity_{R} \otimes \BB(\HH_{\Lambda \setminus R})$. Moreover, $O$ also commutes with every
    \[  \sum_{\omega \in \Omega} \omega Q(\omega) = [H_\Lambda, Q_{R}]. \]
    In particular, considering all local unitaries $Q_{R} = U_{R}$, we conclude that $O$ commutes with 
    \[ [H_\Lambda, U_{R}]U_{R}^{\dagger} = H_\Lambda - U_{R}H_\Lambda U_{R}^{\dagger}. \]
    Integrating with respect to the Haar measure over the unitaries, we have that $O$ commutes with
    \[ H_\Lambda - \int_{\mathcal{U}(R)} U_{R} H_\Lambda U_{R}^{\dagger} dU_{R} \]
    where the integral is with respect to the Haar measure on the local unitaries. This integration yields the normalized partial trace over the region $R$, 
    \begin{equation*}
        H_\Lambda - \int_{\mathcal{U}(R)} U_{R} H_\Lambda U_{R}^{\dagger} dU_{R} = H_\Lambda - \Tr_R(H_\Lambda)\otimes \frac{\identity_R}{d_R} = H_R - \Tr_R(H_R)\otimes \frac{\identity_R}{d_R} \cong H_R \ ,
    \end{equation*}
    where the equivalence is by Lemma~\ref{lemma:traceoperator}.
    Since $O$ thus commutes with $H_R$ and all elements of $\BB(\HH_R)$, it also commutes with 
    \begin{align*}
        H_R^\str \coloneqq \sum_{s\in \strset_R} A_s &\cong \frac{1}{|G|^R} \sum_{\mathbf{g}\in G^R} L^{\mathbf{g}} H_R L^{\overline{\mathbf{g}}} \\
        H_R^\plq \coloneqq \sum_{p\in \plqset_R} B_p &\cong \frac{1}{|\hat{G}|^R} \sum_{\boldsymbol{\chi}\in \hat{G}^R} M^{\boldsymbol{\chi}} H_R M^{\overline{\boldsymbol{\chi}}} 
    \end{align*}
    where we abbreviated
    \begin{align*}
        L^{\mathbf{g}} = \prod_{e\in R} L_e^{g(e)} \quad \mathrm{and} \quad M^{\boldsymbol{\chi}} = \prod_{e\in R} M_e^{\boldsymbol{\chi}(e)} \ .
    \end{align*}
    This is easy to check using the orthogonality relation \eqref{equa:CharacterOrthogonal}.

    Since both cases are equivalent, we will only consider the case of star operators for the rest of this proof.
    As $O$ commutes with $H_R^\str$ and all elements of $\BB(\HH_R)$, it also commutes with any partial trace $\Tr_U(H_R^\str)$ for any subset $U\subseteq R$. 
    We claim that for all $s\in \strset_R$
    \begin{equation}\label{eq:partition_telescope}
         A_s \cong \sum_{r\subseteq \cobdy s\cap R}(-1)^{|\cobdy s\cap R|-|r|} \frac{1}{ d_{R\setminus r}}\Tr_{R\setminus r}(H^\str_R)\ .
    \end{equation}

    To show this equivalence, let us introduce some short-hand notations:  For any $T\subseteq \strset_R$ let $H(T) \coloneqq \sum_{s\in T} A_s$ by slight abuse of notation. In particular, $H^\str_R=H(\strset_R)$.
    Then, for $U\subseteq R$ and $T\subseteq \strset_R$
    \begin{equation}
        \frac{1}{ d_{U}}\Tr_U(H(T)) =  \frac{1}{ d_{U}}\Tr_U(H(T\cap\strset_U)) +  \frac{1}{ d_{U}}\Tr_U(H(T\setminus\strset_U)) \cong  H(T\setminus \strset_U)
    \end{equation}
    since $\Tr_U(A_s)\propto \identity$ for any $s\in \strset_U$.
    Applying this to $H_R^\str$ yields 
    \begin{equation*}
    \frac{1}{ d_{R\setminus r}} \Tr_{R\setminus r}(H^\str_R) \cong H(\strset_R\setminus\strset_{R\setminus r})  =  H(\{s\in \strset_R| \cobdy s\cap R \subseteq r\}) \, .
    \end{equation*}
    The sum on the right-hand side of  \eqref{eq:partition_telescope} can thus be written as
    \begin{align*}
        \sum_{r\subseteq \cobdy s\cap R}(-1)^{|\cobdy s\cap R|-|r|} \frac{1}{ d_{R\setminus r}}\Tr_{R\setminus r}(H_R^\str) 
        &\cong  \sum_{r\subseteq \cobdy s\cap R} (-1)^{|\cobdy s\cap R|-|r|} \sum_{s'\in\strset_R, \cobdy s'\cap R\subseteq r} A_{s'}\\
        &=  \sum_{s'\in\strset_R, \cobdy s'\cap R = \cobdy s \cap R} A_{s'} \ ,
    \end{align*}
    where the last equality is by the inclusion-exclusion principle.
    It remains to show that the sum on the last line equals $A_s$.
    Recall that $R$ is star- and plaquette-connected and has size at least $|R|\geq 2$. Since no two different stars share more than a single edge, the latter implies $|\strset_R|>2$.
    Now, assume that there are two different stars $s_1, s_2 \in \strset_R$ satisfying $\cobdy s_1\cap R = \cobdy s_2\cap R$. 
    As they are in $\strset_R$, their overlap with $R$ cannot be empty. Thus,
    \begin{equation*}
        |\cobdy s_1\cap R |= |\cobdy s_2\cap R|=1 \ .
    \end{equation*}
    This would imply, however, that $s_1$ and $s_2$ are not connected to any other stars in $\strset_R$, contradicting the assumption of a connected region.
\end{proof}

\section{Conditional expectations}\label{sec:condexp}

For the convenience of the reader, let us recall the definition of a conditional expectation:
\begin{definition}
  Let $\NN\subseteq \BB(\HH_{\lsym})$ be a von Neumann subalgebra. A completely positive, unital map $E_{\NN}:\BB(\HH_{\lsym})\to \NN$ is called a conditional expectation  onto $\NN$ if
  \begin{align}\label{eq:condexpdef}
      \forall l, r \in \NN,  O \in \BB(\HH_{\lsym})&\colon \ E_{\NN}(l O r)= l E_{\NN}(O)r \ .
  \end{align}
\end{definition}
The infinite-time limits of the Markovian semigroup generated by a Lindbladian are conditional expectations, which satisfy
\begin{proposition}\label{prop:lindbladianprojector}
    The infinite time limit $\displaystyle\E^\beta_R =\lim_{t\to\infty} e^{t\LL_R^{(\beta)}}$ at inverse temperature $\beta\in [0,\infty)$ has the following properties:
    \begin{enumerate}
        \item It is unital,  completely positive, self-adjoint with respect to $\langle\cdot, \cdot \rangle_{s,\rho}$ for all $s\in [0,1]$, 
        \item It is a conditional expectation onto the von Neumann subalgebra $\ker{\LL_R}$. In particular, it is an orthogonal projection.
  \end{enumerate}
\end{proposition}
These properties are well-known and follow from the self-adjointness of the Lindbladian together with the unitality and complete positivity of the semigroup.

We can derive the following useful expressions for the infinite- and finite-temperature cases.
\begin{theorem}\label{theorem:kernelcondexp}
    Let $R$ be a rectangle. The infinite-temperature conditional expectation onto $\ker(\LL_R)$ is given by
    \begin{equation}\label{eq:infiniteTcondexp}
        \E_R^{0} = \lim_{t\to\infty} e^{t\LL_R^{(0)}} = \PP^\str_R \circ \PP_R^\plq\circ \frac{\Tr_R}{d_R}
    \end{equation}
    where $\PP_R^\str = \prod_{s\in \strset_R}\PP_s^\str$ and $\PP^\plq_R=\prod_{p\in \plqset_R}\PP_p^\plq$ are products of star and plaquette pinchings:
    \begin{align*}
        \PP^\str_s = \prod_{\chi\in\hat{G}} \PP(A_s(\chi)) \quad \mathrm{and} \quad \PP^\plq_p  = \prod_{h\in G} \PP(B_p(h)) ,
    \end{align*}
    with the single operator pinching
    \begin{equation*}
        \PP(P): O\mapsto POP + (\identity-P)O(\identity-P)
    \end{equation*}
    for any star or plaquette operator $P$.
\end{theorem}
We give the proof at the end of the section. The key idea is to show that both sides of the equality are orthogonal projections with the same image and hence equal. 
To extend the conditional expectation to finite temperatures, one replaces the hidden infinite-temperature Gibbs state $\identity/d_R$ with the marginal Gibbs state $\hat{\rho}_R$:
\begin{corollary}\label{cor:finiteTcondexp}
    The finite-temperature conditional expectation is given by
    \begin{equation}\label{eq:finiteTcondexp}
        \E^\beta_R(O) = d_R \E_R^0(\hat{\rho}_R^{1-s} O \hat{\rho}^{s}_R)
    \end{equation}
    for any $O\in \BB(\HH_\Lambda)$ and $s\in [0,1]$, where 
    \begin{equation}
        \hat{\rho}_R = e^{-\beta H_\lsym}(\Tr_R(e^{-\beta H_\lsym}))^{-1} = e^{-\beta H_\lsym}(d_R\E_R^0(e^{-\beta H_\lsym}))^{-1}
    \end{equation}
    is the marginal Gibbs state on $R$.
\end{corollary}
\begin{proof}%[Proof of \cref{cor:finiteTcondexp}]
    Recall that the finite-temperature conditional expectation is a projection onto $\ker(\LL_R^{(\beta)})$ and self-adjoint with respect to $\langle \cdot, \cdot\rangle_{s,\rho}$ for any $s\in [0,1]$. 
    We fix $s\in [0,1]$ and abbreviate by $\hat\E $ the right-hand side of \eqref{eq:finiteTcondexp} with this $s$. We will show that $\hat \E$ is also a projection onto $\ker(\LL_R)$ and self-adjoint with respect to $\langle \cdot, \cdot\rangle_{s,\rho}$. Thus, by the uniqueness of orthogonal projection, it is equal to $ \E_R^\beta $, and hence independent of the choice of $s$.
    
    We start with the description of the image of $\hat\E$. Recall that the kernel of $\LL_R^{(\beta)}$ for $\beta \in [0,\infty)$ is given by the commutant of the jump operators and hence independent of the temperature. It then follows from \cref{theorem:kernelcondexp} that
    \begin{equation*}
        \Img{\hat\E} = \Img{\E_R^0(d_R \hat{\rho}_R^{1-s} \cdot  \hat{\rho}^{s}_R) } \subseteq \Img{\E_R^0}=\ker{\LL_R^{(0)}} = \ker{\LL_R^{(\beta)}} \ .
    \end{equation*}  To prove the reverse inclusion, we pick $O\in\ker{\LL_R^{(\beta)}}$. By \cref{theorem:kernelcondexp}, $O$ is a fixed point of $\E_R^0$, such that
    \begin{equation}
        O = \E_R^0(O) = \hat\E(d_R^{-1}\hat{\rho}_R^{s-1} O \hat{\rho}^{-s}_R) \in \Img\hat\E \ .
    \end{equation}
    
    To show that $\hat \E$ is a projection, we use that by Proposition~\ref{prop:marginaldoublecommutant} for any $O\in \BB(\HH_\Lambda)$ 
    \begin{equation*}
        [\E_R^0(O),\hat{\rho}_R]=0 .
    \end{equation*}
    Together with \eqref{eq:infiniteTcondexp} and the fact that $\E_R^0$ is a conditional expectation, this implies that
    \begin{align*}
        d_R \E_R^0(\hat{\rho}_R^{1-s} d_R \E_R^0(\hat{\rho}_R^{1-s} O \hat{\rho}^{s}_R) \hat{\rho}^{s}_R) 
        & = d_R \E_R^0(\hat{\rho}_R d_R \E_R^0 (\hat{\rho}_R^{1-s}O \hat{\rho}^{s}_R)) \\
        & = d_R \E_R^0(\hat{\rho}_R) d_R \E_R^0 (\hat{\rho}_R^{1-s}O \hat{\rho}^{s}_R) \\
        & = d_R \E_R^0 (\hat{\rho}_R^{1-s}O \hat{\rho}^{s}_R) ,
    \end{align*}
    where the last step follows from  $\E_R^0(\hat{\rho}_R)= \Tr_R(\hat{\rho}_R)/d_R = 1/d_R$.

    To show self-adjointness of $\hat \E$, we compute for  $X,Y\in \BB(\HH_\Lambda)$ the scalar product
    \begin{align*}
        \langle X , \hat{\E}(Y)\rangle_{s, \rho} 
        %= \Tr(\hat\E(X)\rho^{1-s} Y^\dag \rho^{s}) 
        &= d_R\Tr(\rho^{s} X^\dag \rho^{1-s}\E_R^0(\hat{\rho}_R^{1-s} Y \hat{\rho}_R^{s}) ) \\
        &= d_R\Tr( \E_R^0(\rho^{s} X^\dag \rho^{1-s} ) \hat{\rho}_R^{1-s} Y \hat{\rho}_R^{s}) \ ,
    \end{align*}
    where we used the self-adjointness of $\E_R^0$ with respect to the Hilbert-Schmidt inner product.
    For the next step, we expand the Gibbs states
    \begin{equation*}
        \rho^{s}= \frac{e^{-s \beta (H_\lsym-H_R)}e^{- s\beta H_R}}{Z^s} 
        \quad \mathrm{and} \quad 
        \hat{\rho}_R^s = \frac{e^{-s\beta H_\lsym}}{(d_R\E_R^0(e^{-\beta H_\lsym}))^s}  = \frac{e^{-s\beta H_R}}{(d_R\E_R^0(e^{-\beta H_R}))^s} .
    \end{equation*}
    Here the second equality holds since $e^{- s \beta(H_\lsym-H_R)}$ is not supported on $R$ and commutes with all star and plaquette operators. The two kinds of Gibbs states can then be exchanged through the action of $\E_R^0$:
    \begin{equation*}
         \Tr( \hat{\rho}_R^{s}\E_R^0(\rho^{s} X^\dag \rho^{1-s})\hat{\rho}_R^{1-s} Y) = 
          \Tr( \rho^{s}\E_R^0(\hat{\rho}_R^{s} X^\dag \hat{\rho}_R^{1-s})\rho^{1-s} Y) \ .
    \end{equation*}
    Finally, since the Gibbs states are self-adjoint, we find 
    \begin{equation*}
        d_R\E_R^0(\hat{\rho}_R^{s} X^\dag \hat{\rho}_R^{1-s}) = \hat\E(X)^\dag
    \end{equation*}
    concluding the proof.
\end{proof}

%\LG{A very minor comment: a $\rho$-preserving conditional expectation is automatically $\rho$-GNS-symmetric,which can simplifies the self-adjoint argument a little bit. See Takesaki 1972 paper on conditional expectation. SST+SW: I would not change the proof, but feel free to add a remark that one could also use Takesaki, if you know which theorem etc. it is. }

Before presenting the proof of \cref{theorem:kernelcondexp}, we spell out some properties of the star and plaquette pinchings.
\begin{lemma}\label{lemma:pinchingcondexp}
    The star and plaquette pinchings $\PP^\str_s$ and $\PP^\plq_p$ are conditional expectations onto $\{A_s(\chi)\}'_{\chi\in \hat{G}}$ and $\{B_p(h)\}'_{h\in G}$, respectively. Furthermore, they are self-adjoint with respect to the Hilbert-Schmidt inner product.
\end{lemma}
\begin{proof}
    We present the star case, the plaquette case is analogous. 
    The pinching $\PP(A_s(\chi))$ is completely positive, since it is in Kraus form, and unital  since $A_s(\chi)$ is a projection. This implies the unitality and positivity of $\PP_s^\str$. 

    Let $l,r\in \{A_s(\chi)\}'_{\chi\in \hat{G}}$ and $O\in \BB(\HH_\lsym)$. Since $l,r$ commute with all $A_s(\chi)$, we have
    \begin{equation*}
        \PP^\str_s(lOr) = l\PP^\str_s(O)r .
    \end{equation*}
    In particular, $\{A_s(\chi)\}'_{\chi\in \hat{G}} \subseteq \Img \PP_s^\str$.
    For the reverse direction $\supseteq$, note that for all $O \in \BB(\HH_\lsym)$ and all $A=A_s(\chi)$ we have:
    \begin{align*}
        A\PP(A)(O) = AOA = \PP(A)(O)A \ .
    \end{align*}
    Since all $A_s(\chi)$ commute, and so do their single operator-pinchings $\PP(A_s(\chi))$, the above relation carries over to $\PP_s^\str $. Thus  $\Img \PP_s^\str \subseteq \{A_s(\chi)\}'_{\chi\in \hat{G}} $.
    
    Self-adjointness of $ \PP_s^\str $ can be easily checked for the single operator pinching. Let $P=P^\dag=P^2$ be an orthogonal projection in $\BB(\HH)$ (i.e.\ a star or a plaquette operator). Abbreviating $P^\perp=\identity-P$, we may write for any $X,Y\in \BB(\HH)$
    \begin{equation*}
        \Tr(\PP(P)(X)^\dag Y) = \Tr((PX^\dag P + P^\perp X^\dag P^\perp)Y) = \Tr(X^\dag \PP(P)(Y))
    \end{equation*}
    by the cyclicity of the trace.
\end{proof}

\begin{proof}[Proof of \cref{theorem:kernelcondexp}]
    The core of the argument is again the uniqueness of orthogonal projections. By \cref{prop:lindbladianprojector}, the left-hand side of \eqref{eq:infiniteTcondexp} is an orthogonal projection with respect to the Hilbert-Schmidt inner product onto $\ker(\LL_R^{(0)})$. Showing that the right-hand side is also an orthogonal projection with respect to the same inner product and onto the same image then shows the claim. 
    Since any rectangle is star- and plaquette-connected and $|R|>1$ by assumption, $\ker(\LL_R^{(0)})$ is given explicitly by \cref{theorem:kernelNice} and \cref{cor:kernelNice2}. That is, it contains all operators that have no support on $R$ and commute with all stars and plaquettes connected to $R$. 
    
    Note that the product of commuting conditional expectations with respect to the Hilbert-Schmidt inner product is a conditional expectation onto the intersection of their images. 
    
    Also recall that $d_e^{-1} \Tr_e$ is a conditional expectation onto $\identity_e\otimes \BB(\HH_{R\setminus \{e\}})$ and that it is self-adjoint with respect to the Hilbert-Schmidt inner product. The same holds for the pinchings $\PP^\str_s$ and $\PP^\plq_p$ by \cref{lemma:pinchingcondexp}.
    The intersection of all their images is equal to $\ker(\LL_R^{(0)})$.
    
    It remains to be shown that $\Tr_e$, $\PP^\str_s$ and $\PP^\plq_p$ commute pairwise for all $e$,$s$ and $p$.
    For two traces, the commutation is trivial.
    For any pair of star and/or plaquette pinchings, the commutation follows from the fact that stars and plaquettes all commute pairwise.
    The non-trivial pairings are thus a trace with either a star or a plaquette pinching. Since both are equivalent, we will present the star case.
    Let $O\in \BB(\HH_\Lambda)$ and pick $s\in \strset$ and $e\in \cobdy s$ (the case $e\notin\cobdy s$ is trivial).
    Note that the partial trace can be written as the product of two twirls:
    \begin{equation*}
        \frac{1}{d_e}\Tr_e(O) = \frac{1}{|G||\hat{G}|}\sum_{g\in G}\sum_{\chi\in\hat{G}} L_e^g M_e^\chi O M_e^{\overline{\chi}} L_e^{\overline{g}}
    \end{equation*}
    Since $A_s(\chi)$ commutes with any $L_e^g$, only the $M^\chi_e$ twirl needs checking.
    We assume $e\in \cobdy_+s$, the other orientation follows by exchanging $\xi\leftrightarrow\overline{\xi}$ in the argument of $A_s$.
    Fix $\chi,\xi\in \hat{G}$, conjugating $\PP(A_s(\chi))$ with $M_e^\xi$ yields:
    \begin{align*}
        & M_e^{\xi} A_s(\chi)OA_s(\chi) M_e^{\overline{\xi}} + M_e^{\xi} (\identity-A_s(\chi))O (\identity-A_s(\chi)) M_e^{\overline{\xi}} \\
        &= A_s(\chi\xi) M_e^{\xi} O M_e^{\overline{\xi}} A_s(\chi\xi) + (\identity-A_s(\chi\xi)) M_e^{\xi} O M_e^{\overline{\xi}} (\identity-A_s(\chi\xi)) \ .
    \end{align*}
    The star pinching $\PP_s^\str$ commutes with conjugation by $M_e^\chi$, since all single star operator pinchings commute and $\chi\mapsto \xi\chi$ is an isomorphism of $\hat{G}$. In particular, it commutes with the $M^\chi_e$ twirl, concluding the proof.
\end{proof}

\section{Strong martingale condition}\label{sec:strongmartingale}
In this section, we assemble the DS-condition and the expressions derived in the previous sections to show that the conditional expectations on overlapping rectangles factorize approximately. Since we want to find a bound on the MLSI constant that is uniform in system size, we also require the same from the approximate factorization. We call the resulting condition the strong martingale condition.
\begin{definition}\label{def:martingaleCondition}
%The conditional expectations of a family $\FF$ of double models at inverse temperature $ \beta \in [0,\infty) $ is said to satisfy the \emph{strong martingale condition} if there exist constants $K, \xi, L_0>0$  independent of the system size, $\lsym_0\in\FF$, \todo{this sentence does not make sense: clarify 1. what is independent of $\Lambda_0 $ 2. $\geq $ should probably be subset relation...} such that for any $\lsym \geq \lsym_0$ and any pair of overlapping rectangles $UV \subseteq \lsym$ and $VW \subseteq \lsym$ with $L_0\leq\diam(W) \leq 2 \dist(U,W)^2$ and $\diam_{-}(U), \diam_{-}(W)\geq1$ it holds that
The conditional expectations of a family $\FF$ of double models at inverse temperature $ \beta \in [0,\infty) $ are said to satisfy the \emph{strong martingale condition} if there exist constants $K, \xi, L_0>0$ and a minimal $\lsym_0\in\FF$, such that for any system $\lsym\in \FF$ with $ \lsym\supseteq \lsym_0$ and any pair of overlapping rectangles $UV \subseteq \lsym$ and $VW \subseteq \lsym$ with $L_0\leq\diam(W) \leq 2 \dist(U,W)^2$ and $\diam_{-}(U), \diam_{-}(W)\geq1$ it holds that
\begin{equation}\label{equa:martingaleCondition} 
\left(1+K e^{-\xi \dist(U,W)}\right)^{-1}  \E_{UV}^{\beta}\circ  \E_{VW}^{\beta}\leq  \E_{UVW}^{\beta} \leq \left(1+K e^{-\xi \dist(U,W)}\right) \E_{UV}^{\beta}\circ\E_{VW}^{\beta}
\end{equation}  
in the sense of complete positivity order.
\end{definition}
Note that the conditions on the diameters of the sets $U$, $V$, and $W$ make for a weaker statement, which is sufficient for the multi-scale analysis done in \cite{stengele_ModifiedlogarithmicSobolev_2025}.
At infinite temperature, the conditional expectations satisfy an even stronger result. They factorize exactly:
\begin{lemma}\label{lemma:factorization}
    Let $R, R'\subseteq\lsym$ be two (not necessarily disjoint) rectangles.
    Then, the infinite-temperature conditional expectations, as well as the star and plaquette pinchings, factorize exactly:
    \begin{equation}
        \E_{RR'}^0 = \E_{R}^0\circ \E_{R'}^0 \ , \quad \PP_{RR'}^\str = \PP^\str_{R}\circ \PP_{R'}^\str \quad  \mathrm{and} \quad \PP_{RR'}^\plq = \PP^\plq_{R}\circ \PP_{R'}^\plq \ ,
    \end{equation}
    where $RR'$ represents the union of $R$ and $R'$.
\end{lemma}
Since $R$ and $R'$ are arbitrary, this also implies that $\E_R^0$ and $\E_{R'}^0$ commute.
\begin{proof}
    We start by showing that the pinchings factorize. Without loss of generality, we consider $\PP^\str$. Since $\PP^\str_s$ are commuting projections and 
   $ 
        \PP_{R}^\str = \prod_{s\in \strset_R} \PP_{s}^\str
   $, 
    it suffices to check that $\strset_{R}\cup\strset_{R'} = \strset_{RR'}$. This, however, follows by definition. A star $s\in \strset$ is supported on $R$ or $R'$ if and only if it is supported on $R\cup R' = RR'$.

    Using that all pinchings, and the trace, commute, we find:
    \begin{align*}
        \E_{R}^0\circ\E_{R'}^0 &= \PP^\str_{R} \circ \PP_{R}^\plq\circ \frac{\Tr_{R}}{d_{R}} \circ \PP^\str_{R'} \circ \PP_{R'}^\plq\circ \frac{\Tr_{R'}}{d_{R'}}\\
        &= \PP^\str_{RR'} \circ \PP^\plq_{RR'} \circ \frac{\Tr_{R\cap R'}}{d_{R\cap R'}}\circ \frac{\Tr_{RR'}}{d_{RR'}} \ .
    \end{align*}
    Using that the partial trace is a conditional expectation, we write \begin{equation*}
    \Tr_{R\cap R'}(\Tr_{RR'}(O))=\Tr_{R\cap R'}(\identity)\Tr_{RR'}(O)\ ,\  \frac{\Tr_{R\cap R'}}{d_{R\cap R'}}\circ \frac{\Tr_{RR'}}{d_{RR'}}=\frac{\Tr_{RR'}}{d_{RR'}}
    \end{equation*}
    which completes the proof.
    
\end{proof}

For finite temperatures, the conditional expectations then satisfy the strong martingale condition. The main assumption, besides the exact factorization at infinite temperature, is the DS-condition. This can be understood as the approximate factorization of Gibbs states which are associated with $\E_{UV}^\beta\circ \E_{VW}^\beta$ and $\E_{UVW}^\beta$.

\begin{theorem}\label{lemma:condexpordering}
    Assume that a family $\FF$ of double models satisfies the DS-condition at some inverse temperature $\beta$. Then, the conditional expectations of this family satisfy the strong martingale condition.
\end{theorem}
\begin{proof}
    Let $d_0$, $\lsym_0'$ be the minimal distance and system, respectively, given by the DS-condition. Recall that $d_0\geq 2$ and  
    fix $L_0\geq 2d_0^2$ and $\lsym\supseteq \lsym_0'$. Pick rectangles $UV \subseteq \lsym$ and $VW \subseteq \lsym$ as in \cref{def:martingaleCondition}.
    
    To simplify the proof, we will abbreviate $\PP_R \coloneqq \PP_{R}^\str\circ\PP_R^\plq$.
    We make use of the following two properties of the Gibbs state  $\hat{\rho}_R$ on a rectangle. For any star- and plaquette-connected sets $R, R' \subseteq \lsym$
    \begin{enumerate}
        %\item $\Tr_{R}\hat{\rho}_{RR'} = \frac{\hat{\rho}_{RR'}}{\hat{\rho}_{R}} =  \frac{{\color{red} d_{R}}\E_{R}^0e^{-\beta H_{RR'}}}{{\color{red} d_{RR'}}\E_{RR'}^0e^{-\beta H_{RR'}}}$
        \item $\Tr_{R}\hat{\rho}_{RR'} = \frac{\hat{\rho}_{RR'}}{\hat{\rho}_{R}} =  \frac{\Tr_R e^{-\beta H_{RR'}}}{\Tr_{RR'}e^{-\beta H_{RR'}}}$.
        \item $\hat{\rho}_R$, any of its partial traces over star- and plaquette-connected regions, and thus also their inverses, commute with all star and plaquette operators.
    \end{enumerate}
    The first item follows from a simple computation, which uses~\eqref{eq:rhohatdef}. 
    For the second item, note that by the first item, it is sufficient to check it for $\hat{\rho}_R$. Expanding it, we find
    \begin{equation*}
        \hat{\rho}_R = e^{-\beta H_R} \left(\Tr_R (e^{-\beta H_R})\right)^{-1} \ .
    \end{equation*}
    The first factor commutes with all stars and plaquettes. 
    For the second factor, we use \cref{prop:marginaldoublecommutant}, which ensures that the second factor lies in the bicommutant of all stars and plaquettes. However, since all stars and plaquettes commute, they lie in their own commutant. Thus, they commute with any element of their bicommutant.

    Note that $V$ is not a rectangle, but the overlap of two rectangles. Nonetheless, it is either star- and plaquette-connected, or consists of two star- and plaquette-connected components, $V=V_1V_2$. Each connected component of $V$ has a width of at least $2$, by assumption.
    Since the inner diameter of $U$ and $W$ is at least $1$, we have $\dist(V_1,V_2)\geq2$ and thus $\strset_{V_1}\cap \strset_{V_2}=\plqset_{V_1}\cap \plqset_{V_2}=\emptyset$. 
    For traces of Gibbs states, we can apply the above fact twice. Since $\hat{\rho}_{V_1}$ has no support on $V_2$ and vice versa, $\hat{\rho}_{V}$ factorizes, and we find
    \begin{equation*}
        \Tr_{V_1V_2} \hat{\rho}_{UV} = \Tr_{V_1}\left( \frac{\hat{\rho}_{UV}}{\hat{\rho}_{V_1}}\right) = \Tr_{V_1}\left( \hat{\rho}_{UV} \right)\hat{\rho}_{V_1}^{-1} = \frac{\hat{\rho}_{UV}}{\hat{\rho}_{V}} \ ,
    \end{equation*}
    which commutes with all stars and plaquettes by the second fact above.
    
    We can use the facts above and that the conditional expectations and their constituents are projections as well as \cref{lemma:factorization} and \cref{cor:finiteTcondexp} to rewrite     
    \begin{align*}
        \E^\beta_{UVW} (O)
        &= d_{UVW} \E_{UVW}^0 ( \hat{\rho}_{UVW}O) \\
        &= d_{VW} \Tr_{U} \circ\PP_U  \circ\E_{VW}^0 \left( \frac{\hat{\rho}_{UVW}}{\hat{\rho}_{VW}}\hat{\rho}_{VW}O\right) \\
        &= d_{VW} \Tr_{U} \circ\PP_U  \circ\E_{VW}^0 \left(\Tr_{VW}(\hat{\rho}_{UVW})\hat{\rho}_{VW}O\right) \\
        &= \Tr_{U} \left( \Tr_{VW}(\hat{\rho}_{UVW}) \PP_U  \circ\E^{\beta}_{VW} \left(O\right)\right) \ .
    \end{align*}
    Here we used that $ d_{UVW} = d_U d_{VW} $ in the second step, and that $\Tr_{VW}(\hat{\rho}_{UVW})$ is not supported on $VW$ in the fourth step.
    Similarly, we can rewrite
    \begin{align*}
         \E^\beta_{UV} \circ \E^\beta_{VW}(O) 
         &= d_{UV} \E_{UV}^0\left(\hat{\rho}_{UV} \E^\beta_{VW}(O) \right)\\
         &=\Tr_{U}\circ\PP_{UV}\circ\Tr_{V}\left(\hat{\rho}_{UV} \E^\beta_{VW}(O) \right)\\
         &= \Tr_U\left(\Tr_{V}(\hat{\rho}_{UV}) \PP_{UV}\circ\E^\beta_{VW}(O) \right)\\    
         &= \Tr_U\left(\Tr_{V}(\hat{\rho}_{UV}) \PP_{U}\circ\E^\beta_{VW}(O) \right)\ .
    \end{align*}
    In the last step, we used the fact that $\PP_{UV} \circ \PP_{VW}= \PP_{U}\circ \PP_{VW}$ to absorb the pinching on $V$ in $\E_{VW}^{\beta}$.
    %\LG{ in the last equality, $\PP_{UV}$ becomes $\PP_{U}$ because...SST: I added a comment}
    To conclude the proof note that $\PP_{U}\circ\E^\beta_{VW}$ and $\Tr_{U}$ are completely positive. By \cref{prop:marginaldoublecommutant}, the states $\Tr_{VW}(\hat{\rho}_{UVW})$ and $\Tr_{V}(\hat{\rho}_{UV})$ commute with the image of $\PP_{U}\circ\E^\beta_{VW}$. 
    Using, in addition, that by assumption
    \begin{equation*}
        (1-Ke^{-\xi\dist(U,W)})\Tr_{V}\hat{\rho}_{UV} \leq \Tr_{VW}\hat{\rho}_{UVW} \leq (1+Ke^{-\xi\dist(U,W)})\Tr_{V}\hat{\rho}_{UV}
    \end{equation*}
    we find the same ordering of the conditional expectations
    \begin{equation*}
        (1-Ke^{-\xi\dist(U,W)} ) \E^\beta_{UV} \circ \E^\beta_{VW} \leq \E^\beta_{UVW} \leq (1+Ke^{-\xi\dist(U,W)} ) \E^\beta_{UV} \circ \E^\beta_{VW} \ ,
    \end{equation*}
    as claimed.
\end{proof}

\section{MLSI through multiscale analysis}\label{sec:MLSI}
This section contains the lines of reasoning in proving that the strong martingale condition implies a uniformly bounded MLSI constant. We highlight the important steps and ideas and refer to~\cite{stengele_ModifiedlogarithmicSobolev_2025} for details.
The proof is based on a widely used tactic both for classical \cite{stroock_equivalencelogarithmicSobolev_1992,martinelli_ApproachEquilibriumGlauber_1994, martinelli_ApproachEquilibriumGlauber_1994a} and quantum \cite{kastoryano_QuantumGibbsSamplers_2016} semigroups. 
The key idea is to define an MLSI constant on scale $L$ by
\begin{equation*}
    \alpha(L) := \min_{R \in \RR_L} \alpha(\LL_{R})
\end{equation*}
where $\RR_L$ is the set of all rectangles of size at most $L$.
Under suitable conditions, we can lower bound $\alpha(2L)$ in terms of $\alpha(L)$. Iterating this procedure yields a uniform lower bound on $\alpha(L)$ in terms of $\alpha(L_0)$ for some value $L_0$ that is independent of system size.

This can be formalized as follows:
\begin{theorem}
If the family $\FF$ satisfies the strong martingale condition, then there exists $L_{0} = L_{0}(K, \xi)>0$ independent of the system size such that for every $L \geq L_{0}$
\[ \alpha(L) \geq \exp\left( \frac{-120}{L_{0}^{1/3}}\right) \alpha(L_{0}) \ . \]
\end{theorem}

The main ideas of the proof are as follows.  First, we fix some $L_{0} \geq 1$ such that for every $L \geq L_{0}$ the strong martingale condition holds on scale $\sqrt{L}$ with
\[ Ke^{-\xi \sqrt{L}} < \frac{1}{28}. \]
Using the strong martingale condition and arguing as in \cite[Sec. 5.2]{stengele_ModifiedlogarithmicSobolev_2025}, we can find that
\[ D(\sigma \| \E_{UVW}^*(\sigma)) \leq (1+\varepsilon_{L}) \left( D(\sigma \| \E_{VW}^*(\sigma) ) + D(\sigma \| \E_{UV}^*(\sigma) )\right). \]
This approximate factorization of the conditional expectation can then be used as in \cite[Prop. 6.1]{stengele_ModifiedlogarithmicSobolev_2025} to conclude that for all $L \geq L_{0}$
\[ \alpha(2L) \geq \left(1+\frac{24}{L^{1/3}} \right)^{-1} \alpha(L). \]
This is achieved using a \emph{divide and conquer method}. The key step is to decompose each rectangle $R$ into two overlapping  subrectangles $UV$ and $VW$, with intersecting rectangle $V$, and to estimate from below $\alpha(\LL_{R})$ in terms of $\alpha(\LL_{UV})$ and $\alpha(\LL_{VW})$. 
Iterating this process as in \cite[Eqs. (6.2) and (6.3)]{stengele_ModifiedlogarithmicSobolev_2025} yields the desired inequality.

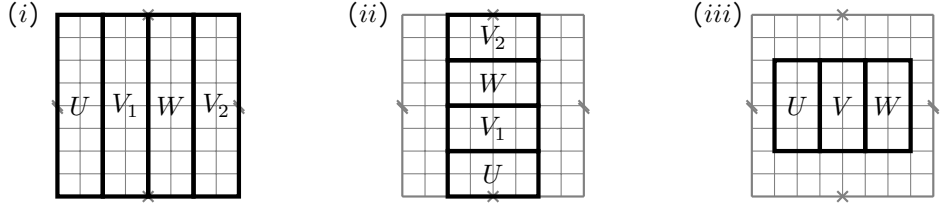
\begin{figure}[ht]
\centering
\begin{tikzpicture}[equation,scale=0.3]
    \node at (-1.5, 8) {$(i)$};
    %\filldraw[red!20] (0,0) rectangle (2,8);
    %\filldraw[blue!20] (4,0) rectangle (6,8);
    \draw[gray, thin] (0,0) grid (8,8);
    \draw[gray,thick, postaction=torus horizontal] (0,0) -- (8,0);
    \draw[gray,thick, postaction=torus horizontal] (0,8) -- (8,8);
    \draw[gray,thick, postaction=torus vertical] (0,0) -- (0,8);
    \draw[gray,thick, postaction=torus vertical] (8,0) -- (8,8);
    \draw[black, ultra thick] (0,0) rectangle (2,8);
    \node at (1, 4) {$U$};
    \draw[black, ultra thick] (2,0) rectangle (4,8);
    \node at (3, 4) {$V_{1}$};
    \draw[black, ultra thick] (4,0) rectangle (6,8);
    \node at (5, 4) {$W$};
    \draw[black, ultra thick] (6,0) rectangle (8,8);
    \node at (7, 4) {$V_{2}$};
    \end{tikzpicture}
    %%%%%%%%%%%%%%%%%%%%%%%%%%%
    \hspace{1cm} %%%%%%%%%%%%%%%%%%%%%%%%%%%%%%
\begin{tikzpicture}[equation,scale=0.3]
    \node at (-1.5, 8) {$(ii)$};
    %\filldraw[yellow!20] (0,0) rectangle (8,2);
    %\filldraw[yellow!20] (0,6) rectangle (8,8);
    %\filldraw[red!20] (0,2) rectangle (2,6);
    %\filldraw[blue!20] (4,2) rectangle (6,6);
    % \draw[gray, thin] (0,0) grid (8,8);
    % \draw[gray,thick, postaction=torus horizontal] (0,0) -- (8,0);
    % \draw[gray,thick, postaction=torus horizontal] (0,8) -- (8,8);
    % \draw[gray,thick, postaction=torus vertical] (0,0) -- (0,8);
    % \draw[gray,thick, postaction=torus vertical] (8,0) -- (8,8);
    % \draw[black, ultra thick] (0,2) rectangle (2,6);
    % \node at (1, 4) {$U$};
    % \draw[black, ultra thick] (2,2) rectangle (4,6);
    % \node at (3, 4) {$V_{1}$};
    % \draw[black, ultra thick] (4,2) rectangle (6,6);
    % \node at (5, 4) {$W$};
    % \draw[black, ultra thick] (6,2) rectangle (8,6);
    % \node at (7, 4) {$V_{2}$};
    %\node at (4, 7) {$D$};
    \draw[gray, thin] (0,0) grid (8,8);
    \draw[gray,thick, postaction=torus horizontal] (0,0) -- (8,0);
    \draw[gray,thick, postaction=torus horizontal] (0,8) -- (8,8);
    \draw[gray,thick, postaction=torus vertical] (0,0) -- (0,8);
    \draw[gray,thick, postaction=torus vertical] (8,0) -- (8,8);
    \draw[black, ultra thick] (2,0) rectangle (6,2);
    \node at (4, 1) {$U$};
    \draw[black, ultra thick] (2,2) rectangle (6,4);
    \node at (4, 3) {$V_{1}$};
    \draw[black, ultra thick] (2,4) rectangle (6,6);
    \node at (4, 5) {$W$};
    \draw[black, ultra thick] (2,6) rectangle (6,8);
    \node at (4, 7) {$V_{2}$};
\end{tikzpicture}
    %%%%%%%%%%%%%%%%%%%%%%%%%%%
    \hspace{1cm} %%%%%%%%%%%%%%%%%%%%%%%%%%%%%%
\begin{tikzpicture}[equation,scale=0.3]
    \node at (-1.5, 8) {$(iii)$};
    %\filldraw[yellow!20] (0,0) rectangle (8,8);
    %\filldraw[white] (1,2) rectangle (7,6);
    %\filldraw[red!20] (1,2) rectangle (3,6);
    %\filldraw[blue!20] (5,2) rectangle (7,6);
    \draw[gray, thin] (0,0) grid (8,8);
    \draw[gray,thick, postaction=torus horizontal] (0,0) -- (8,0);
    \draw[gray,thick, postaction=torus horizontal] (0,8) -- (8,8);
    \draw[gray,thick, postaction=torus vertical] (0,0) -- (0,8);
    \draw[gray,thick, postaction=torus vertical] (8,0) -- (8,8);
    \draw[black, ultra thick] (1,2) rectangle (3,6);
    \node at (2, 4) {$U$};
    \draw[black, ultra thick] (3,2) rectangle (5,6);
    \node at (4, 4) {$V$};
    \draw[black, ultra thick] (5,2) rectangle (7,6);
    \node at (6, 4) {$W$};
    %\node at (4, 7) {$D$};
    \end{tikzpicture}
    \caption{Three sets of overlapping rectangles. (i) The full torus is split into $V_2UV_1$ and $V_1WV_2$. Both rectangles still wrap around the torus. (ii) Splitting a rectangle that wraps around in one direction yields two new rectangles, which overlap on both ends. (iii) A rectangle that does not wrap around splits into two smaller ones with only one connected component in the overlap. }
    \label{fig:decompositionTorusQDM}
\end{figure}

It remains to show that $\alpha(L_0)$ admits a lower bound $\alpha_0>0$ that is independent of the system size. 
This part is spelled out in \cite[Sec. 6.1]{stengele_ModifiedlogarithmicSobolev_2025} and boils down to:
\begin{proposition}
    Let $R\subset\lsym$ be a rectangle of size at most $L_0$. Let $\LL_R$ be a Davies Lindbladian for $H_R$ with local, translation invariant,  jump operators. Then, there exists a constant $C>0$ independent of the system size such that the MLSI constant of $\LL_R$ satisfies $\alpha(\LL_R)>C$.
\end{proposition}
The proof is based on a general lower bound on the MLSI constant derived in \cite{gao_Completeentropicinequalities_2022}. 
The essential ingredient is a lower bound on the spectral gap of $\LL_R$.
As shown in \cite{stengele_ModifiedlogarithmicSobolev_2025}, the uniform lower bound $g$ on the jump rates implies a lower bound $\gap(\LL_R)\geq\gap(\tilde{\LL}_R)$, where $\tilde{\LL}_R$ is the Davies Lindbladian with jump rates $h_{e,i}(\omega) = g e^{\beta \omega}$, independent of $i$ and $e$.
Then, by translation invariance of the jump operators, the infimum $\inf_{R}\gap(\tilde{\LL}_R)$ simplifies to a minimum, since there are only finitely many sizes of rectangles below a given size.

{\small
\paragraph{Acknowledgements.} 
		This work was supported by the DFG under grant TRR 352--Project-ID 470903074 (A.\,C.,C.\,R.,S.\,S.,S.\,W.). S.\,W. was also supported under EXC-2111 -- 390814868.
 A.\,L.~acknowledges support from the Italian Ministry of University and Research (MUR), through ``Programma per Giovani Ricercatori Rita Levi Montalcini'', the grant ``Dipartimento di Eccellenza 2023-2027'' of Dipartimento di Matematica, Politecnico di Milano, as well as the National Group of Mathematical Physics (GNFM) of the Italian Institute for High Mathematics (INdAM).
 L.\,G.~acknowledges
  support from the National Natural Science Foundation of China (grant No.~12401163), the Department of Science and Technology of Hubei Province  (Project No. 2025EHA041, Project No. 2025AFA044).\\
 D.\,P-G. and A.\,P-H.~acknowledge financial support from the following grants PID2020-113523GB-I00 and PID2023-146758NB-I00  funded by MICIU/AEI/10.13039/501100011033. D.\,P-G. acknowledges support from grant  TEC-2024/COM-84-QUITEMAD-CM, funded by Comunidad de Madrid, and grant CEX2023-001347-S, funded by MICIU/AEI/10.13039/501100011033. This work has been financially supported by the Ministry for Digital Transformation and the Civil Service of the Spanish Government through the QUANTUM ENIA project call – Quantum Spain project, and by the European Union through the Recovery, Transformation and Resilience Plan – NextGenerationEU within the framework of the Digital Spain 2026 Agenda. This project was funded within the QuantERA II Programme which has received funding from the EU’s H2020 research and innovation programme under the GA No 101017733.
}

        \bibliographystyle{abbrvurl}
        \bibliography{lit}

        \end{document}